\documentclass[acmsmall,10pt]{acmart}

\usepackage{booktabs}
\usepackage{amsmath}
\usepackage{amssymb}
\usepackage{amsthm}
\usepackage{wasysym}
\usepackage{stmaryrd}
\usepackage{color}
\usepackage{enumerate}
\usepackage{multirow}
\usepackage{diagbox}
\usepackage{xspace}
\usepackage{algorithm2e}
\usepackage{url}
\usepackage[tight,footnotesize]{subfigure}
\usepackage{graphicx}
\graphicspath{{figs/}}
\DeclareGraphicsExtensions{.pdf,.jpeg,.png}

\newcommand{\oomit}[1]{}
\newcommand{\zap}[1]{}

\newtheorem{rrule}{Rule}

\newcommand{\ifrule}[2]{\begin{array}{cc}#1\\\hline#2\end{array}}

\newcommand{\tsc}{scheduling constraint\xspace}

\newcommand{\rr}   {\mathtt{rr}\xspace}
\newcommand{\ww}   {\mathtt{ww}\xspace}

\newcommand{\vars} {\mathbb{V}\xspace}
\newcommand{\events}{\mathbb{E}\xspace}
\newcommand{\ceevents}{\mathbb{E}_{\ce}\xspace}
\newcommand{\stms}{\mathbb{T}\xspace}
\newcommand{\conflicts}{\mathbb{A}\xspace}
\newcommand{\krset}{\mathbb{O}\xspace}

\newcommand{\ppprec}{\prec_{P}^0\xspace}

\newcommand{\rfprecc}{{\lhd}\xspace}
\newcommand{\cepoprec}{{\prec_{\ce}^0}\xspace}
\newcommand{\cerfprec}{{\lhd_{\ce}}\xspace}
\newcommand{\ceprec}{\prec_\ce\xspace}

\newcommand{\ce}{\pi\xspace}
\newcommand{\order}{\lambda\xspace}
\newcommand{\execution}{\tau\xspace}

\newcommand{\stm}{t\xspace}
\newcommand{\kr}{o\xspace}

\newcommand{\stmo}{\Lambda\xspace}

\newcommand{\monolithic}{\alpha\xspace}
\newcommand{\krf}{\varpi\xspace}

\newcommand{\true}{TRUE\xspace}
\newcommand{\sat}{SAT\xspace}
\newcommand{\smt}{SMT\xspace}

\newcommand{\eog}{EOG\xspace}
\newcommand{\cegar}{CEGAR\xspace}
\newcommand{\pth}{PThreads\xspace}
\newcommand{\svcomp}{SV-COMP\xspace}
\newcommand{\svcompc}{SV-COMP 2017\xspace}

\newcommand{\minisat}{\textsc{MiniSat2}\xspace}
\newcommand{\gcbmc}{\textsc{Yogar-CBMC}\xspace}
\newcommand{\cbmc}{\textsc{CBMC}\xspace}
\newcommand{\threader}{\textsc{Threader}\xspace}
\newcommand{\lazycseq}{\textsc{Lazy-CSeq-Abs}\xspace}
\newcommand{\mucseq}{\textsc{MU-CSeq}\xspace}

\DeclareMathOperator{\var}{var}
\DeclareMathOperator{\type}{type}
\DeclareMathOperator{\guard}{guard}
\DeclareMathOperator{\clk}{clk}
\DeclareMathOperator{\sel}{sel}

\setcopyright{none}             
\begin{document}

\title{Scheduling Constraint Based Abstraction Refinement for Multi-Threaded Program Verification}

\author{Liangze Yin}
\affiliation{School of Computer, National University of Defense Technology, China}
\email{yinliangze@163.com}
\author{Wei Dong}
\affiliation{School of Computer, National University of Defense Technology, China}
\email{wdong@nudt.edu.cn}
\author{Wanwei Liu}
\affiliation{School of Computer, National University of Defense Technology, China}
\email{wwliu@nudt.edu.cn}
\author{Ji Wang}
\affiliation{School of Computer, National University of Defense Technology, China}
\email{wj@nudt.edu.cn}

\thanks{This work has been submitted to the IEEE for possible publication. Copyright may be transferred without notice,
after which this version may no longer be accessible.}

\begin{abstract}
Bounded model checking is among the most efficient techniques for the automatic verification of concurrent programs.
However, encoding all possible interleavings often requires a huge and complex formula, which significantly limits the salability.
This paper proposes a novel and efficient abstraction refinement method for multi-threaded program verification.
Observing that the huge formula is usually dominated by the exact encoding of the scheduling constraint, this paper proposes a \tsc based abstraction refinement method, which avoids the huge and complex encoding of BMC.
In addition, to obtain an effective refinement, we have devised two graph-based algorithms over event order graph for counterexample validation and refinement generation, which can always obtain a small yet effective refinement constraint.
Enhanced by two constraint-based algorithms for counterexample validation and refinement generation, we have proved that our method is sound and complete w.r.t. the given loop unwinding depth.
Experimental results on \svcompc benchmarks indicate that our method is promising and significantly outperforms the existing state-of-the-art tools.
\end{abstract}



\maketitle

\section{Introduction}
Facilitated by the popularization of multi-core architectures, concurrent programs are becoming popular to take full advantage of the available computing resources.
However, due to the nondeterministic thread interleavings, traditional approaches such as testing and simulation are hard to guarantee the correctness of such programs.
Automatic program verification has become an important complementary to traditional approaches.
Given that most of the errors can be detected with small loop unwinding depths, bounded model checking (BMC) has been proven to be one of the most efficient techniques for the automatic verification of concurrent programs \cite{QadeerR05, Beyer17}.
However, due to the complex inter-thread communication, a huge encoding is usually required to offer an exact description of the concurrent behavior, which greatly limits the scalability of BMC for concurrent programs.

This paper focuses on multi-threaded programs based on shared variables and \emph{sequential consistency (SC)} \cite{AdveG96}.
For these programs, we have observed that the \emph{\tsc}, which defines that ``for any pair $\langle w, r \rangle$ s.t. $r$ reads the value of a variable $v$ written by $w$, there should be no other write of $v$ between them,'' significantly contributes to the complexity of the behavior encoding.
In the existing work of BMC, to encode the \tsc, each access of a shared variable is associated with a ``clock variable''. The scheduling constraint is then encoded into a complicated logic formula over the state and clock variables, the size of which is cubic in the number of shared memory accesses \cite{AlglaveKT13}.

Inspired by this observation, this paper proposes a novel method for multi-threaded program verification which performs abstraction refinement by weakening and strengthening the \tsc.
It initially ignores the \tsc and then obtains an over-approximation abstraction of the original program (w.r.t. the given loop unwinding depth).
If the property is safe on the abstraction, then it also holds on the original program.
Otherwise, a counterexample is obtained and the abstraction is refined if the counterexample is infeasible.
N. Sinha and C. Wang also performed abstraction refinement to deal with the overhead of an exact encoding of the concurrent behavior \cite{SinhaW11}.
However, their abstraction model was performed by restricting the sets of read events and read-write links, while we consider all read events and read-write links but ignore the \tsc.

The efficiency of our method depends on the number of iterations required to verify the property and the sizes of the constraints added during the refinement process.
Another innovation of this paper is that, to verify the property with a small number of small problems, we have devised two graph-based algorithms over event order graph (\eog) for counterexample validation and refinement generation, s.t. an effective refinement constraint can be obtained in each refinement iteration.
Whenever an abstraction counterexample is determined to be infeasible, we can always obtain a set of ``core kernel reasons'' of the infeasibility, which can usually be encoded into simple constraints and reduce a large amount of space.
In our experiments, most of the programs can be verified within dozens of refinement iterations.
Meanwhile, the increased size of the abstraction during the refinement process can usually be ignored compared with that of the initial abstraction.

Our graph-based \eog validation method is effective in practice.
Given an infeasible \eog, it can usually identify the infeasibility with rare exceptions.
If it is not sure whether an \eog is feasible or not, we explore a constraint-based \eog validation process to further validate its feasibility.
If an infeasibility is returned, we explore a constraint-based refinement generation process to refine the abstraction.
Enhanced by these two constraint-based processes, we have proved that our method is sound and complete w.r.t the given loop unwinding depth.

We have implemented the proposed method on top of \cbmc and applied it to the benchmarks in the concurrency track of \svcompc \cite{sv-comp2017}. The experimental results demonstrate that our method drastically improves the verification performance.
Without the scheduling constraint, the formula size reduces to 1/8 on average, whereas the number of CNF clauses increased during the refinement process can usually be ignored compared with that of the abstraction.
Moreover, our tool has successfully verified all these examples within 1550 seconds and 43 GB of memory.
By contrast, \lazycseq | a leading tool for concurrent program verification | spended 9820 seconds and 104 GB memory to achieve the same score.
Our tool has won the gold medal in the concurrency track of \svcompc \cite{sv-comp2017} (Warning: It will violate our anonymity).

The contributions of this paper are listed as follows.

\begin{enumerate}
 \item This paper presents a \tsc based abstraction refinement method for multi-threaded program verification, which avoids the huge and complex constraint to encode the concurrent behavior.
 \item This paper presents two graph-based algorithms over event order graph for counterexample validation and refinement generation, which can always obtain a small yet effective refinement constraint in practice.
 \item To ensure the soundness, we have enhanced our method by two constraint-based algorithms for counterexample validation and refinement generation. In this manner, a both efficient and sound method for multi-threaded program verification is obtained.
 \item We have implemented our method on top of \cbmc. The evaluation on the \svcompc benchmarks indicates that our method is promising and significantly outperforms the existing state-of-the-art tools.
\end{enumerate}

The rest of this paper is organized as follows.
Section~\ref{sec:preliminaries} introduces the preliminaries.
Section~\ref{sec:illustration} outlines and illustrates our proposed method by presenting a running example.
Sections~\ref{sec:validation} and~\ref{sec:refinement} present our \eog-based counterexample validation and refinement generation algorithms, respectively.
Section~\ref{sec:soundness} discusses the soundness and efficiency of our method.
Section~\ref{sec:experiment} provides the experimental results.
Section~\ref{sec:relatedWork} reviews the related work, and Section~\ref{sec:conclusion} concludes the paper.

\section{Preliminaries}
\label{sec:preliminaries}

\subsection{Multi-Threaded Program}
A \emph{multi-threaded} program $P$ consists of $N \geq 1$ concurrent
threads $P_i$ ($1 \leq i \le N$).
It contains a set of variables which are partitioned into \emph{shared variables} and \emph{local variables}.
Each thread $P_i$ can read/write both the shared variables and its local variables.
We focus on programs based on \pth, one of the most popular libraries for multi-threaded programming.
It uses \texttt{pthread\_create(\&\texttt{t}, \&attrib, \&f, \&args)} to create a new thread \texttt{t}, and \texttt{pthread\_join(\texttt{t}, \&\_return)} to suspend the current thread until thread \texttt{t} terminates~\footnote{More information about \pth can be found at https://computing.llnl.gov/tutorials/pthreads/}.

In this paper, we assume each variable access is atomic.
We also assume that each statement $\stm$ of a multi-threaded program is either 1) a \emph{global statement} that contains only one shared variable access (it may further contain multiple local variable accesses) or 2) a \emph{local statement} that only operates on local variables.
A statement with multiple shared variable accesses can always be translated to a sequence of global statements.
By defining the expressions suitably and using source-to-source transformations, we can model all statements using global and local statements.

We also assume that all functions are inlined and all loops are unwound by a limited depth (it is a basic proviso in BMC).
We also omit the discussion on modeling the sophisticated C language data elements, such as pointers, structures, arrays, and heaps, etc., because they are irrelative with the concurrency and we deal with them in the same way as CBMC does.
The discussion on PThread primitives, such as \texttt{pthread\_mutex\_lock} and \texttt{pthread\_mutex\_unlock}, etc., are also omitted in this paper.
In CBMC, they are implemented by shared variables and we deal with them in the same way as CBMC does.

Given a multi-threaded program, we write $\vars$ for the set of shared variables.
An \emph{event} $e$ is a read/write access to a shared variable.
We use $\events$ to denote all of them.
Each global statement corresponds to an event, i.e., the event contained in the global statement.
Each $e\in\events$ is associated with an element $\var(e)\in\vars$, a type $\type(e)$, and a literal $\guard(e)$,
which represent the accessed variable, the type of access, and the guard condition literal, respectively.
$\type(e)$ can be either $\rr$ (i.e., ``read'') or $\ww$ (i.e., ``write'').
Any event $e$ with $\var(e) = v$ and $\type(e)=\rr$ (resp. $\type(e)=\ww$) is called a read (resp. a write) of $v$.
To express the execution orders of different events, we also associate each event with an unique natural number $\clk(e)$.
$\clk(e_1) < \clk(e_2)$ represents that $e_1$ executes before $e_2$.

The program $P$ determines a partial order $\ppprec\subset\events\times\events$.
Intuitively, $e_1\ppprec e_2$ (or we write $(e_1, e_2) \in \ppprec$) indicates that ``$e_1$ should happen before
$e_2$ according to the \emph{program order} of $P$''.
According to the program order, $e_1\ppprec e_2$ holds in the following three cases.
\begin{itemize}
\item For any two events $e_1$ and $e_2$ of the same thread, if $e_1$
  must happen before $e_2$ according to the sequential
  semantics, then we have $e_1\ppprec e_2$.
\item If \texttt{pthread\_create()} is used to create a new thread
  \texttt{t} at some point $p$ of the current thread, then for any event
  $e_1$ of the current thread before $p$ and $e_2$ of the thread \texttt{t},
  we have $e_1\ppprec e_2$.
\item If \texttt{pthread\_join(\texttt{t})} is used to suspend the current
  thread at some point $p$, then for any event $e_1$ of the thread
  \texttt{t} and $e_2$ of the current thread after $p$, we have
  $e_1\ppprec e_2$.
\end{itemize}

A read-write link $(e_1, e_2)$ represents that ``$e_2$ reads the value written by $e_1$''.
Therefore, $\type(e_1)=\ww$, $\type(e_2)=\rr$, $\var(e_1)=\var(e_2)$, and the value of $e_2$ is equal to that of $e_1$.
In addition, there should be no other ``write'' of $\var(e_1)$ happening between them.
Given a read-write link $\order := (e_1, e_2)$, we denote by $\sel(\order)$ the \emph{read-write link literal} (a boolean variable) that represents the link.

\subsection{Bounded Model Checking}
Bounded model checking \cite{BiereCCZ99} is one of the most applicable techniques to
alleviate the state space explosion problem in concurrent program verification.
Given that most of the errors can be detected with small loop unwinding depths, the unwinding depth for those loops and recursions is limited \cite{QadeerR05}.
Instead of explicitly enumerating all thread interleavings,
BMC employs a symbolic representation to encode the verification problem, which is then solved by a \sat/\smt solver.
If a positive answer is given, then a satisfying assignment
corresponding to a feasible counterexample is acquired. Otherwise, the
program is proven safe w.r.t. the given loop unwinding depth.

In BMC, the monolithic encoding of a multi-threaded program is usually represented as $\monolithic := \phi_{init} \wedge \rho \wedge \zeta \wedge \xi$, where $\phi_{init}$ is the initial states, $\rho$ encodes each thread in isolation, $\zeta$ formulates that ``each read of a variable $v$ may read the result of any write of $v$'', and $\xi$ formulates the \tsc, which defines that ``for any pair $\langle w, r \rangle$ s.t. $r$ reads the value of a variable $v$ written by $w$, there should be no other write of $v$ between them'' \cite{AlglaveKT13}.

\section{Method Overview and Illustration}
\label{sec:illustration}

\subsection{Method Overview}
The performance of BMC is usually decided by that of the constraint solving, the performance of which depends significantly on the size of the constraint problem.
Hence, an important way to improve the performance of BMC is to reduce the size of the constraint problem.
In BMC of multi-threaded programs, we have observed that the monolithic encoding $\monolithic$ is usually dominated by the \tsc $\xi$.
To reduce the constraint problem size, we propose to ignore the \tsc in the constraint solving process.
An abstraction of the monolithic encoding is then obtained, which is defined as follows.

\begin{definition}
\label{def:abstraction}
Given a multi-threaded program, the abstraction ignoring the \tsc can be formulated as $\varphi_0 := \phi_{init} \wedge \rho \wedge \zeta$, where $\phi_{init}$ is the initial states, $\rho$ encodes each thread in isolation, and $\zeta$ formulates that ``each read of a variable $v$ may read the result of any write of $v$''.
\end{definition}

The \tsc $\xi$ defines a set of order requirements among the events.
All of them should be satisfied for any concrete execution of a multi-threaded program.
Hence, whenever a counterexample $\ce$ of the abstraction is obtained, further validation is required to determine whether $\ce$ satisfies all the order requirements defined in $\xi$.
If that is true, $\ce$ is feasible.
Otherwise, $\ce$ is infeasible.
In other words, it may not correspond to a concrete execution.
In this case, we continually search the rest of the abstraction space for another new counterexample, until a feasible counterexample is found or all counterexamples of the abstraction have been determined to be infeasible.

Fig.~\ref{fig:x} presents an overview of our method.
Given a multi-threaded C program, we first add the abstraction $\varphi_0$ and the error states $\phi_{err}$ to the abstraction model.
If it is unsatisfiable, then the property is proven safe w.r.t. the given loop unwinding depth.
Otherwise, a counterexample of the abstraction is provided.
Given that the scheduling constraint is ignored in the abstraction, this counterexample may be infeasible and further validation is required.
In our method, the feasibility of an abstraction counterexample is determined via validating the feasibility of its corresponding event order graph (\eog).
An intuitive method for \eog validation is constraint solving.
If the \eog is infeasible, then the abstraction is refined by exploring the unsatisfiable core.
However, this method is not effective for refinement generation (cf. Section~\ref{sec:refinement}).
To obtain an effective refinement, we have devised two graph-based algorithms over \eog for \eog validation and refinement generation, in which a small yet effective refinement can always be obtained if the \eog is determined to be infeasible.
However, this method is not complete. It can only give an infeasible answer (cf. Section~\ref{sec:validation}).
To make our method both efficient and sound, we first adopt the graph-based \eog validation method.
If the \eog is determined to be infeasible, the graph-based refinement process is performed to obtain an effective refinement constraint.
Otherwise, we employ the constraint-based validation process to further validate the \eog.
Our experiments demonstrate that our graph-based \eog validation method is effective in practice.
It can always identify the infeasibility of an infeasible \eog with rare exceptions.
Actually, the constraint-based refinement generation process (the dashed part of Fig.~\ref{fig:x}) has never been invoked on \svcompc benchmarks.

\begin{figure}
\centering
\includegraphics[width=11cm]{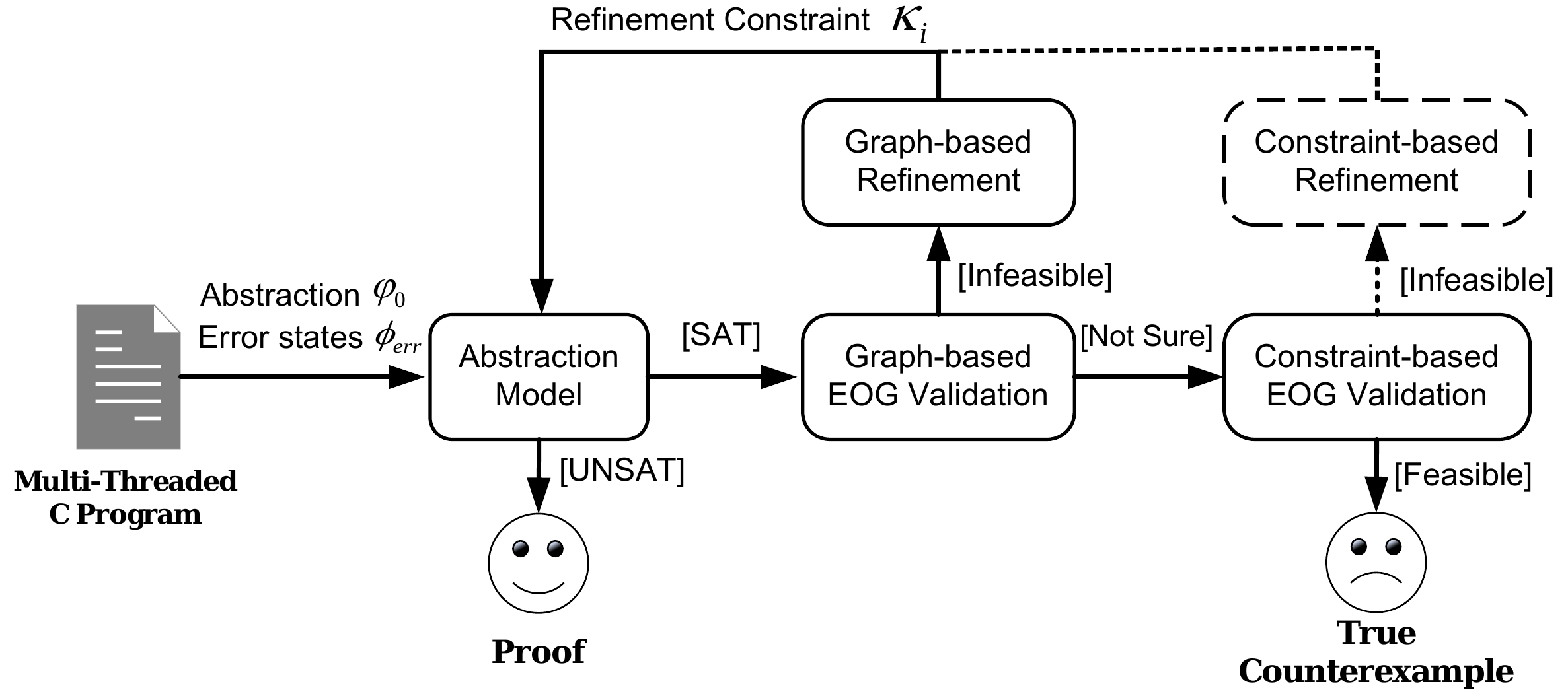}
\caption{An overview of our method}
\label{fig:x}
\end{figure}

\subsection{Method Illustration}
We provide a running example to illustrate our method.  The program involves three threads, namely, \texttt{main},
\texttt{thr1}, and \texttt{thr2}, as shown in Fig.~\ref{fig:2}(a).
In this example, the set of shared variables $\vars := \{x, y, m, n\}$, which are initialized to $\{1, 1, 0, 0\}$, respectively.  The \texttt{main} thread creates threads \texttt{thr1} and \texttt{thr2}, and
then waits until these two threads terminate.
We attempt to verify that it is impossible for both $m$ and
$n$ to be $1$ after the exit of \texttt{thr1} and \texttt{thr2}.
This program offers a modular proof to this property.

\begin{figure}
  \centering
  \includegraphics[width=9cm]{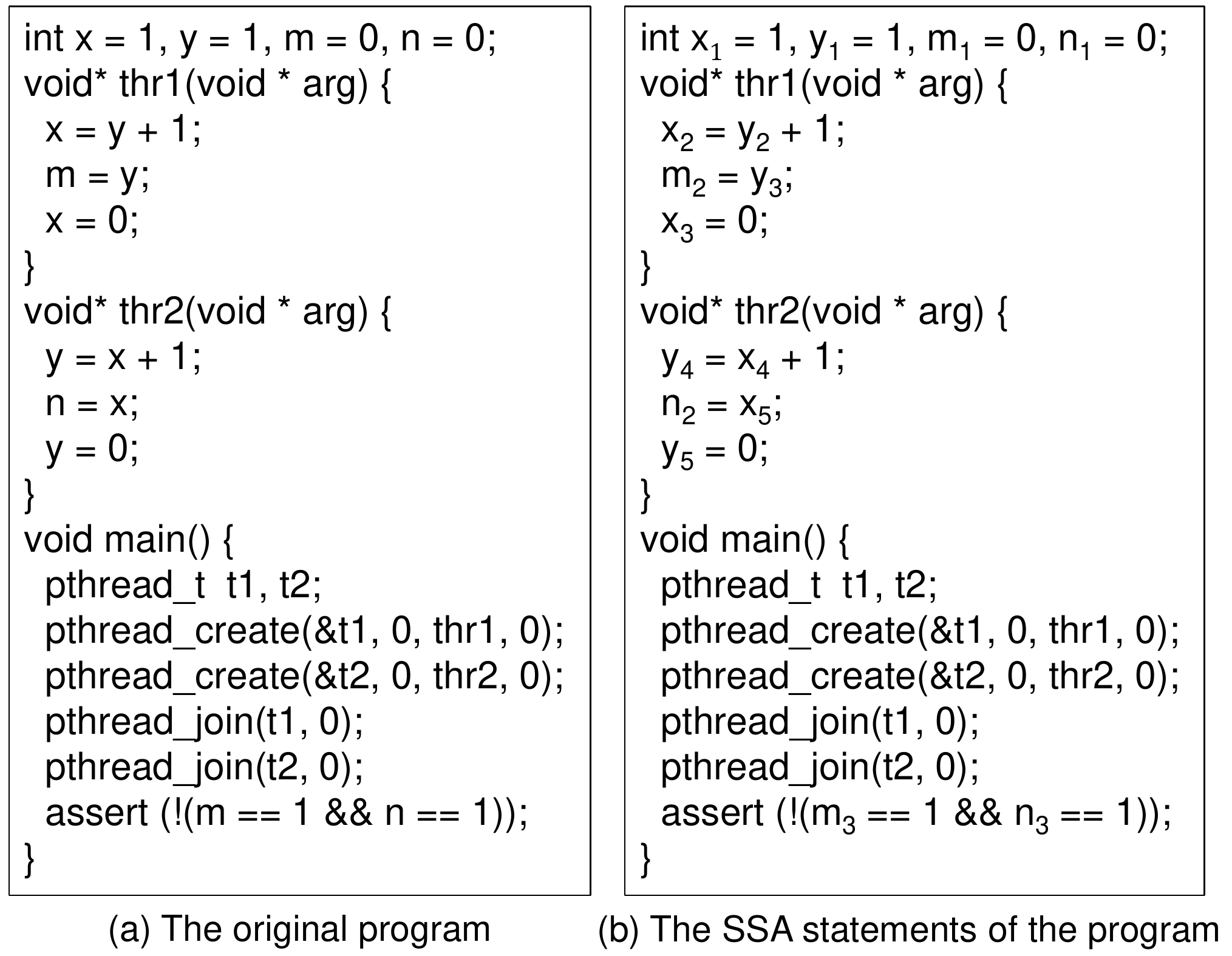}
  \caption{A three-thread program}
  \label{fig:2}
\end{figure}

\paragraph{Initial abstraction.}
We use $\phi_{init}$ and $\phi_{err}$ to denote the initial and
error states, respectively, while $\rho_{main}$, $\rho_{thr1}$, and $\rho_{thr2}$ denote the transition relationships of these three threads, respectively, and $\rho$ is the conjunction of those transition relationships of all threads.

To encode the program, as shown in Fig.~\ref{fig:2}(b), we convert the original program into a set of {\em{static single assignment}} (SSA) statements, in which the program variables are renamed s.t. each variable is assigned only once.
Particularly, each ``read'' of any shared variable also has a unique name.
Then $\phi_{init}$, $\phi_{err}$, and $\rho$ are defined as follows.
Note that the transition relationship of each thread (such as $\rho_{main}$, $\rho_{thr1}$, and $\rho_{thr2}$) encodes that thread in isolation, i.e., it doesn't consider the thread communications.

\begin{equation*}
  \begin{array}{ll}
    \phi_{init} &:= (x_1 = 1) \wedge (y_1 = 1) \wedge (m_1 = 0) \wedge (n_1 = 0) \\
    \phi_{err} &:= (m_3 = 1) \wedge (n_3 = 1) \\
    \rho_{thr1} &:= (x_2 = y_2 + 1) \wedge (m_2 = y_3) \wedge (x_3 = 0) \\
    \rho_{thr2} &:= (y_4 = x_4 + 1) \wedge (n_2 = x_5) \wedge (y_5 = 0) \\
    \rho_{main} &:= true \\
    \rho &:= \rho_{thr1} \wedge \rho_{thr2} \wedge \rho_{main}
  \end{array}
\end{equation*}

To encode $\zeta$, we must identify the behavior of every ``read'' event.
Consider the shared variable $x$ for example.
There are five read/write accesses to the variable $x$.
For each access, as shown in Fig.~\ref{fig:2}(b), we rename $x$ to a unique name in the SSA statements, i.e., $x_1$, $x_2$, $\cdots$, $x_5$.
We denote by $e_{x_i}$ ($1 \leq i \leq 5$) the event corresponding to $x_i$.
Then $\{e_{x_1}, e_{x_2}, e_{x_3}\}$ and $\{e_{x_4}, e_{x_5}\}$ are the sets of
``writes'' and ``reads'' of $x$, respectively.
We use a read-write link literal $s_{v,i,j}$ to indicate that $e_{v_j}$ reads the value written by $e_{v_i}$ ($v\in\vars$).
The encoding $\psi_{x_i}$ ($i = 4, 5$), defined below, indicates that the value of $x_i$ can take any value of $x_1$, $x_2$, and $x_3$.
Given that the variable $x_i$ can not take several different values simultaneously, the formula $s_{x,4,1} \vee s_{x,4,2} \vee s_{x,4,3}$ represents that, among these three literals, there is one and only one true literal.
We denote by $\zeta_x$ the conjunction of $\psi_{x_4}$ and $\psi_{x_5}$.
It formulates the possible behaviors of all ``reads'' of $x$.

\begin{equation*}
  \begin{array}{ll}
    \psi_{x_4} :=
    & (s_{x,4,1} \Rightarrow (x_4 = x_1)) \wedge \\
    & (s_{x,4,2} \Rightarrow (x_4 = x_2)) \wedge \\
    & (s_{x,4,3} \Rightarrow (x_4 = x_3)) \wedge \\
    &(s_{x,4,1} \vee s_{x,4,2} \vee s_{x,4,3}) \\

    \mathcal{\psi}_{x_5} :=
    & (s_{x,5,1} \Rightarrow (x_5 = x_1)) \wedge \\
    & (s_{x,5,2} \Rightarrow (x_5 = x_2)) \wedge \\
    & (s_{x,5,3} \Rightarrow (x_5 = x_3)) \wedge \\
    &(s_{x,5,1} \vee s_{x,5,2} \vee s_{x,5,3}) \\

    \zeta_{x} :=&\psi_{x_4} \wedge \psi_{x_5} \\
  \end{array}
\end{equation*}

Similarly, we obtain the corresponding formulas of $\zeta_{y}$,
$\zeta_{m}$, and $\zeta_{n}$.
Let $\zeta :=\zeta_{x} \wedge \zeta_{y} \wedge \zeta_{m} \wedge \zeta_{n}$.
The initial abstraction can then be formulated as follows.
\begin{equation}
\label{equ:abstract}
  \varphi_0 := \phi_{init}\wedge\rho\wedge\zeta
\end{equation}

\paragraph{Constraint solving of the first round.}
Using $\varphi_0\wedge\phi_{err}$ as input to a constraint solver will return SAT and yield a counterexample $\ce_0$, which is a set of assignments to the variables in $\varphi_0\wedge\phi_{err}$.

\paragraph{Counterexample validation of the first round.}

Given that the \tsc is excluded from the abstraction, such a counterexample may be infeasible.
In our method, a counterexample $\ce$ is validated via validating its corresponding \eog (cf. Section~\ref{sec:validation}), which captures all the order requirements among the events of $\ce$.
We first employ the graph-based \eog validation method to determine its feasibility.

Fig. \ref{fig:3} shows the \eog corresponding to $\ce_0$.
In figures that describe EOGs, the white and gray nodes denote ``writes'' and ``reads'' occurring in the corresponding counterexample, respectively.
A solid arrow with a triangular head from $e_1$ to $e_2$ represents a \emph{program order}, which requires that $e_1$ should happen before $e_2$.
A dashed arrow from $e_1$ to $e_2$ represents a \emph{read-write link} $(e_1, e_2)$.
It requires that 1) $e_1$ should happen before $e_2$, and 2) no ``write'' of $\var(e_1)$ can happen between them.
A solid arrow with a hollow head from $e_1$ to $e_2$ represents a \emph{derived order}, which is derived from existing order requirements. It also requires that $e_1$ should happen before $e_2$.
For brevity, in these figures, we use the subscripts of an event as labels, that is, we use $v_i$ to represent the event $e_{v_i}$.
Now the question is: Is there any total order of all these nodes that satisfies all these order requirements?
This is not a trivial problem.
However, if there exists some cycle in the graph, then the answer must be ``no'', and the counterexample is infeasible.

\begin{figure}
\centering
\includegraphics[width=6.5cm]{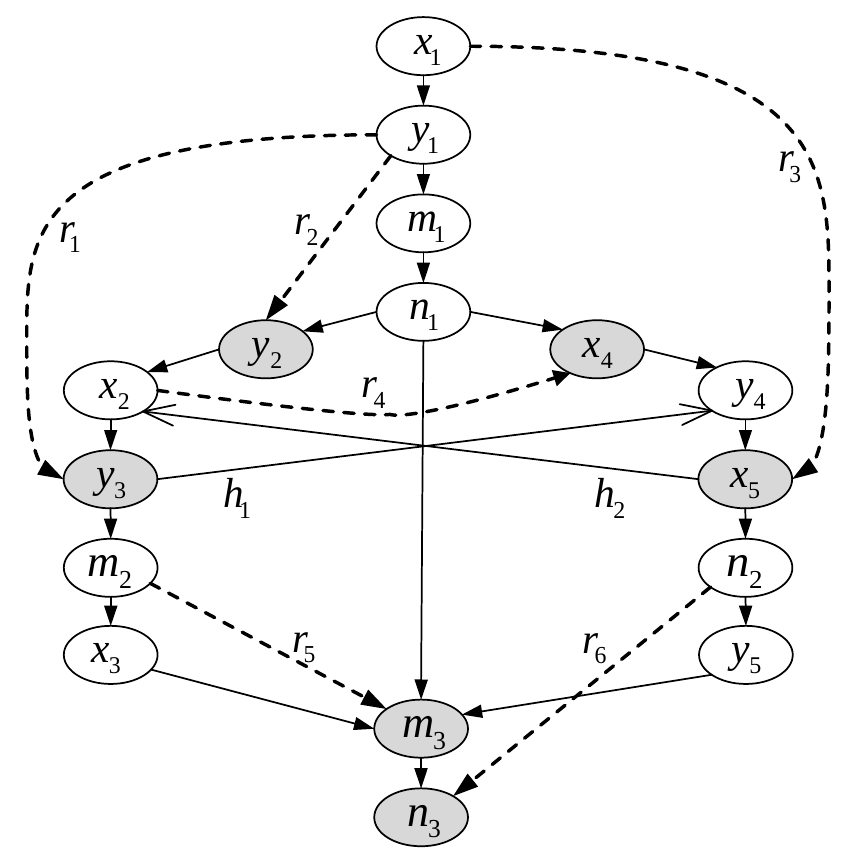}
\caption{EOG of counterexample $\ce_0$. }
\label{fig:3}
\end{figure}

Given a counterexample $\ce$ and two events $e_1$ and $e_2$, we use $e_1 \prec_{\ce} e_2$ to represent that $e_1$ should happen before $e_2$ in $\ce$.
By applying our graph-based \eog validation algorithm (cf. Section \ref{sec:validation}), we can deduce two derived orders $e_{y_3} \prec_{\ce_0} e_{y_4}$ and $e_{x_5} \prec_{\ce_0} e_{x_2}$, which are denoted by $h_1$ and $h_2$ respectively~\footnote{We can deduce more orders from the \eog, but we only list $h_1$ and $h_2$, because they will be used later.}.
Fig.~\ref{fig:3} shows two cycles, including $C_1: e_{x_2} \prec_{\ce_0} e_{y_3} \prec_{\ce_0} e_{y_4} \prec_{\ce_0} e_{x_5} \prec_{\ce_0} e_{x_2}$ and
$C_2: e_{x_2} \prec_{\ce_0} e_{x_4} \prec_{\ce_0} e_{y_4} \prec_{\ce_0} e_{x_5} \prec_{\ce_0} e_{x_2}$.
Therefore, $\ce_0$ is infeasible.

\paragraph{Refinement of the first time.}
To prune more search space rather than just one counterexample, we should find the ``kernel reasons'' that make the counterexample infeasible.
According to our graph-based kernel reason analysis algorithm (cf. Section~\ref{sec:refinement}), the derived orders $h_1$ and $h_2$ are caused by $r_1$ and $r_3$, respectively (according to Rule \ref{rule:3}).
$h_1$ is derived as follows.
According to the order requirements of $r_1$, we have that $y_1 \prec_{\ce_0} y_3$, and no ``write'' of $y$ can be executed between $y_1$ and $y_3$.
According to the program orders, we have $y_1 \prec_{\ce_0} y_4$.
Hence, we can deduce that $y_3 \prec_{\ce_0} y_4$.
Given that the guard conditions for all these events are true, as long as $r_1$ holds in the counterexample, we may obtain the derived order $h_1$.
Hence, the reason of $h_1$ is $r_1$.
Similarly, we can obtain that the reason of $h_2$ is $r_3$.
Given that the guard conditions for all these events are true, the reason for any program order is \true.
And according to Section~\ref{sec:refinement}, the reason for any read-write link is itself.

A kernel reason of a cycle $C$ is a conjunction of those kernel reasons for those orders constructing $C$.
We can obtain that $C_1$ is caused by $r_1 \wedge r_3$, and $C_2$ is caused by $r_3 \wedge r_4$.
Given that $r_1$, $r_3$, and $r_4$ are represented by read-write link literals $s_{y,3,1}$, $s_{x,5,1}$, and $s_{x,4,2}$, respectively, we obtain that the kernel reason of $C_1$ is $\{s_{y,3,1}, s_{x,5,1}\}$, and the kernel reason of $C_2$ is $\{s_{x,5,1}, s_{x,4,2}\}$.
We hence use $\kappa_0 := \neg(s_{y,3,1} \wedge s_{x,5,1}) \wedge\neg(s_{x,5,1} \wedge s_{x,4,2})$ as the refinement constraint, which contains only two simple CNF clauses.

\paragraph{Second constraint solving.}
Let $\varphi_1 := \varphi_0 \wedge \kappa_0$.
When using $\varphi_1 \wedge \phi_{err}$ as input, the constraint solver returns SAT again, and produces a new counterexample $\ce_1$.

\paragraph{Second counterexample validation.}
By applying our graph-based \eog validation algorithm, the \eog corresponding to $\ce_1$ has two cycles.
Hence, $\ce_1$ is also infeasible.

\paragraph{Refinement, the second time.}
Again, we apply our graph-based kernel reason analysis algorithm.
The refinement constraint is formulated as $\kappa_1 := \neg(s_{x,4,3} \wedge s_{x,5,1}) \wedge\neg(s_{y,2,4} \wedge s_{x,4,3}) \wedge \neg(s_{y,4,3} \wedge s_{x,4,3})$, which contains three simple CNF clauses.
Here, one of these two cycles has two different kernel reasons.

\paragraph{Third constraint solving.}
Same as before, let $\varphi_2 := \varphi_1 \wedge \kappa_1$. When using $\varphi_2 \wedge \phi_{err}$ as input to a constraint solver, we get another counterexample $\ce_2$.

\paragraph{Third counterexample validation.}
By applying our graph-based \eog validation algorithm, the \eog corresponding to $\ce_2$ also has two cycles.
Hence, $\ce_2$ is also infeasible.

\paragraph{Refinement, the third time.}
By applying our graph-based kernel reason analysis algorithm, we obtain a new refinement constraint $\kappa_2 := \neg(s_{y,3,1} \wedge s_{y,2,5}) \wedge\neg(s_{x,5,2} \wedge s_{y,2,5})$, which contains two simple CNF clauses.

\paragraph{Constraint solving, the fourth time.}
Let $\varphi_3 := \varphi_2 \wedge \kappa_2$.
When using $\varphi_3 \wedge \phi_{err}$ as input, the constraint solver returns UNSAT this time, which indicates that the property is safe.

Given that the serial of constraint solving queries are incremental, they can be solved in an incremental manner.
From this example, we can observe that:
\begin{enumerate}
  \item Without the \tsc, the size of the abstraction is much smaller than that of the monolithic encoding. In this example, excluding the 3049 CNF clauses encoding the \texttt{pthread\_create} and \texttt{pthread\_join} function calls, the monolithic encoding contains 10214 CNF clauses, while the abstraction $\varphi_0$ contains only 1018 CNF clauses.
  \item With our graph-based refinement generation method, the verification problem can usually be solved with a small number of refinement iterations. In this example, only three refinements are required to verify the property.
  \item With our graph-based refinement generation method, the number of clauses increased during the refinement process can usually be ignored compared with that of the abstraction. In this example, only 7 simple CNF clauses are added during the refinement process, while the initial abstraction contains 4067 CNF clauses.
  \item Our graph-based \eog validation method is effective to identify the infeasibility in practice. In this example, all the four abstractions are infeasible. All of them can be detected by our graph-based \eog validation process. The constraint-based \eog validation and refinement processes have never been invoked.
\end{enumerate}

\section{\eog-Based Counterexample Validation}
\label{sec:validation}

\subsection{Counterexample and Event Order Graph}
\begin{definition}
\label{def:counterexample}
A counterexample $\ce$ of an abstraction $\varphi_i$, or a counterexample $\ce$ for short, is a set of assignments to the variables in $\varphi_i \wedge \phi_{err}$, where $\phi_{err}$ is the error states.
\end{definition}

A counterexample $\ce$ is an execution of the abstraction that falsifies the property.
It defines a trace for each thread and the read-write relationship among the ``reads'' and ``writes'' occurring in $\ce$.
Given that the \tsc is ignored in the abstraction, the counterexample may not be feasible, i.e., it may not correspond to any concrete execution.
Note that the execution order of those statements from different threads is not defined.
If the counterexample is feasible, it may correspond to multiple concrete executions.

Given a counterexample $\ce$, we use
$\ceevents \subseteq \events$ to denote the set of events occurring in
$\ce$.  We define a partial order $\ceprec\subseteq\ceevents\times\ceevents$.  Intuitively,
$e_1\ceprec e_2$ represents that ``$e_1$ should happen before $e_2$ in $\ce$'', i.e., $\clk(e_1) < \clk(e_2)$.
{We also use $\cepoprec$ to denote the restriction of
$\ppprec$ on $\ce$. In this case, $\cepoprec\subseteq\ppprec$}, and we have $\cepoprec \subseteq \ceprec$.
We also focus on the partial order
$\cerfprec$, which is called the \emph{read-from} relationship of $\ce$.
$e_1\cerfprec e_2$ (or we write $(e_1, e_2) \in \cerfprec$) represents that ``$(e_1, e_2)$ is a read-write link in $\ce$''.
According to this definition, we obtain $\cerfprec \subseteq \ceprec$.
An element in $\cepoprec$ (resp. in $\cerfprec$) is called a \emph{program order} (resp. \emph{read-from order}) of $\ce$.

%

A counterexample $\ce$ defines a quadruple $\langle s_{0, \ce}, \stms_\ce, \events_\ce, \rfprecc_\ce\rangle$, where $s_{0, \ce}$ is the initial states, $\stms_\ce$ is the set of statements contained in $\ce$, $\events_\ce$ is the set of events occurring in $\ce$, and $\rfprecc_\ce \subseteq \events_\ce \times \events_\ce$ is the read-from relationship that links each ``read'' $r \in \events_\ce$ to a ``write'' $w \in \events_\ce$ s.t. $r$ reads the value written by $w$.
Note that $\cepoprec$ is the restriction of $\prec_{P}^0$ to $\events_\ce$.
Table~\ref{tab:notation} lists the set of notations we use for a counterexample $\ce$.

\begin{table}
  \centering
  \caption{Notations for a counterexample $\ce$}
  \label{tab:notation}
  \begin{tabular}{|c|l|}
    \hline
    \multicolumn{1}{|c|}{Notation} &
    \multicolumn{1}{|l|}{Meaning} \\

    \hline
    $s_{0, \ce}$ & The initial states of $\ce$. \\
    $\stms_\ce$ & The set of statements contained in $\ce$ \\
    $\ceevents$  & The set of events occurring in $\ce$. \\
    $e_1\ceprec e_2$    & $e_1$ should happen before $e_2$ in $\ce$. \\
    $e_1\cepoprec e_2$    & $e_1\ceprec e_2$ according to the program order.\\
    $e_1\cerfprec e_2$ & $e_2$ reads the value written by $e_1$. \\
    \hline
  \end{tabular}
\end{table}

\begin{definition}
\label{def:ce-feasible}
A counterexample $\ce$ is feasible if a concrete execution $\execution$ of the original program can be constructed from $\ce$.
\end{definition}

To validate a counterexample $\ce$, we define a \emph{concrete execution} of a multi-threaded program as follows.
\begin{definition}
\label{def:execution}
Let $s_0$ be an initial state of $P$, $\stmo := \stm_0 \cdots \stm_n$ be a statement sequence, and $s \stackrel{\stm}{\rightarrow}s'$ indicate that $s'$ is the successor of $s$ by $\stm$.
The tuple $(s_0, \stmo)$ defines a concrete execution of $P$ iff there exists a state sequence $s_0 \cdots s_{n+1}$ s.t. $s_i \stackrel{\stm_i}{\rightarrow}s_{i+1}$ for all $0 \leq i \leq n$.
\end{definition}

According to Definition~\ref{def:execution}, to validate a counterexample $\ce$, one should find a statement sequence of $\stms_\ce$ that defines a concrete execution.
Note that for a concrete execution $\execution$, the events occurring in $\execution$ must satisfy the following order requirements: 1) for each program order $(e_1, e_2)$, we have $e_1$ happens before $e_2$; and 2) for each read-write link $(e_1, e_2)$, we have $e_1$ happens before $e_2$, and no write of $\var(e_1)$ happens between $e_1$ and $e_2$.
Given that those local statements do not affect other threads, the crucial issue to construct a concrete execution is to find a total order $<_\pi$ over $\events_\ce$ s.t. $<_\ce$ obeys all the above order requirements.
To address this problem, we introduce the \emph{event order graph (\eog)} notion to capture all order requirements of a counterexample.

\begin{definition}
\label{def:eog}
Given a counterexample $\ce$, the \emph{event order graph} (\eog) $G_\ce$ is a triple $\langle\events_\ce,\cepoprec,\cerfprec \rangle$, where the nodes are the events in $\events_\ce$, and the edges are the orders defined in $\cepoprec$ and $\cerfprec$.
Each node corresponds to either a ``read'' or a ``write'' of $\ce$, and each edge corresponds to either a program order or a read-from order of $\ce$.
For each edge corresponding to a program order $(e_1, e_2) \in \cepoprec$, it requires that $\clk(e_1) < \clk(e_2)$; and for each edge corresponding to a read-from order $(e_1, e_2) \in \cerfprec$, it requires that $\clk(e_1) < \clk(e_2)$ and $\forall e_3\in \events_\ce, ((\var(e_3) = \var(e_1))\wedge(\type(e_3) = \ww)) \Rightarrow (\clk(e_3) < \clk(e_1) \vee \clk(e_2) < \clk(e_3))$.
\end{definition}

\begin{definition}
\label{def:eog-feasible}
An \eog $G_\ce$ is feasible iff there exists a total order $<_\ce$ over $\events_\ce$ s.t. $<_\ce$ obeys all the order requirements defined in $G_\ce$.
\end{definition}

\begin{theorem}
\label{thm:eog_and_ce}
A counterexample $\ce$ is feasible iff the corresponding \eog $G_\ce$ is feasible.
\end{theorem}
\begin{proof}
\textbf{Sufficiency.} If $\ce$ is feasible, then we can construct a concrete execution $\execution$ from $\ce$. Suppose that the statement sequence of $\execution$ is $\stmo$. We order all the events in $\events_\ce$ as the execution order of those corresponding global statements in $\stmo$, and obtain a total order $<_\ce$ over $\events_\ce$, which is consistent with $\cepoprec$ and $\cerfprec$. Therefore, $G_\ce$ is feasible.

\textbf{Necessity.} We try to construct a concrete execution $\execution$ from $\pi$. If $G_\ce$ is feasible, then there must exist a total order $<_\pi$ over $\events_\ce$ that is consistent with $\cepoprec$ and $\cerfprec$. To obtain a statement sequence $\stmo$ of $\stms_\ce$,
we first order all the global statements in $\stms_\ce$ as the order of those corresponding events in $<_\pi$, and obtain a total order $<_g$ of all the global statements.
Then we ``place'' the local statements in $\stms_\ce$ into $<_g$.
Specifically, we first give a total order $t_0 \cdots t_n$ of all local statements according to the program order.
Afterward, we insert the local statements into $<_g$ according to this order and obtain a statement sequence $\stmo$.
For each local statement $t_i$, suppose that among all its predecessors of global statements (according to the program order), $t_i^g$ is lastly scheduled in $<_g$,
then $t_i$ is scheduled after both $t_{i-1}$ and $t_i^g$.
For instance, if $t_{i-1}$ is scheduled before $t_i^g$, then $t_i$ is scheduled after $t_i^g$.
Otherwise, $t_i$ is scheduled after $t_{i-1}$.
Given that $<_\pi$ is both consistent with $\cepoprec$ and $\cerfprec$, the tuple $(s_{0,\ce}, \stmo)$ is a concrete execution.
Therefore, $\ce$ is feasible.
\end{proof}

Now we ask, how can we validate the feasibility of an \eog?
An intuitive way is to exactly encode all the order requirements defined in Definition~\ref{def:eog} into a constraint.
The \eog is feasible iff the constraint is satisfiable.
However, as we will justify later
(cf. Section~\ref{sec:refinement} and~\ref{sec:experiment}), constraint solving is not effective enough for refinement generation.

\subsection{Graph-Based \eog Validation}
\label{subsec:4.1}

According to Definition~\ref{def:eog-feasible}, any edge $(e_1, e_2)$ of an \eog $G_\ce$ requires that $e_1 \ceprec e_2$.
Hence, an \eog must be infeasible if it contains some cycles.
Note that a read-from order $(e_1, e_2) \in \cerfprec$ further requires that no other write of $\var(e_1)$ could happen between $e_1$ and $e_2$.
Some implicit order requirements deducible from the \eog must exist.
We call them \emph{derived orders} of the \eog.
For each derived order $(e_1, e_2)$, it also requires that $\clk(e_1) < \clk(e_2)$, i.e., $e_1 \ceprec e_2$.

Consider the \eog shown in Fig.~\ref{fig:6}(a), where
$\{e_{x_0}, e_{x_2}\}$ and $\{e_{x_1}\}$ are the ``writes'' and ``reads'' of $x$, respectively.
$e_{x_0} \cepoprec e_{x_1}$, and $e_{x_2} \cerfprec e_{x_1}$.
We deduce that $e_{x_0} \prec_\ce e_{x_2}$ because $e_{x_0} \cepoprec e_{x_1}$ and $e_{x_2} \cerfprec e_{x_1}$, of which
the latter implies that no ``write'' of $x$ can happen between $e_{x_2}$ and $e_{x_1}$.

Based on this observation, we can deduce as many derived orders as possible first, and add them to $\ceprec$. If some cycle exists in $\ceprec$, then the \eog must be infeasible.
To this end, we propose the following three rules to produce derived orders. We initially obtain
$\ceprec := \cepoprec \cup \cerfprec$.

\begin{rrule}
\label{rule:1}
$\ifrule{e_1\ceprec e_2,\  e_2\ceprec e_3}
  {e_1\ceprec e_3}$.
\end{rrule}

Rule \ref{rule:1} only reflects the transitivity of $\prec_\ce$.

\begin{rrule}
\label{rule:2}
$\ifrule{e_1\cerfprec e_2, \  e_3\ceprec e_2, \
  \type(e_3)=\ww, \  \var(e_3)=\var(e_1)}
{e_3\ceprec e_1}$.
\end{rrule}

Given that $e_1 \cerfprec e_2$, $\type(e_3) = \ww$, and $\var(e_3) = \var(e_2)$, then either $\clk(e_3)<\clk(e_1)$ or $\clk(e_2)<\clk(e_3)$.
Given that $e_3 \ceprec e_2$ holds, then $\clk(e_3)<\clk(e_2)$ can be obtained.
Therefore, we have $\clk(e_3)<\clk(e_1)$, which implies $e_3 \ceprec e_1$.

\begin{figure}
\centering
\includegraphics[width=8.5cm]{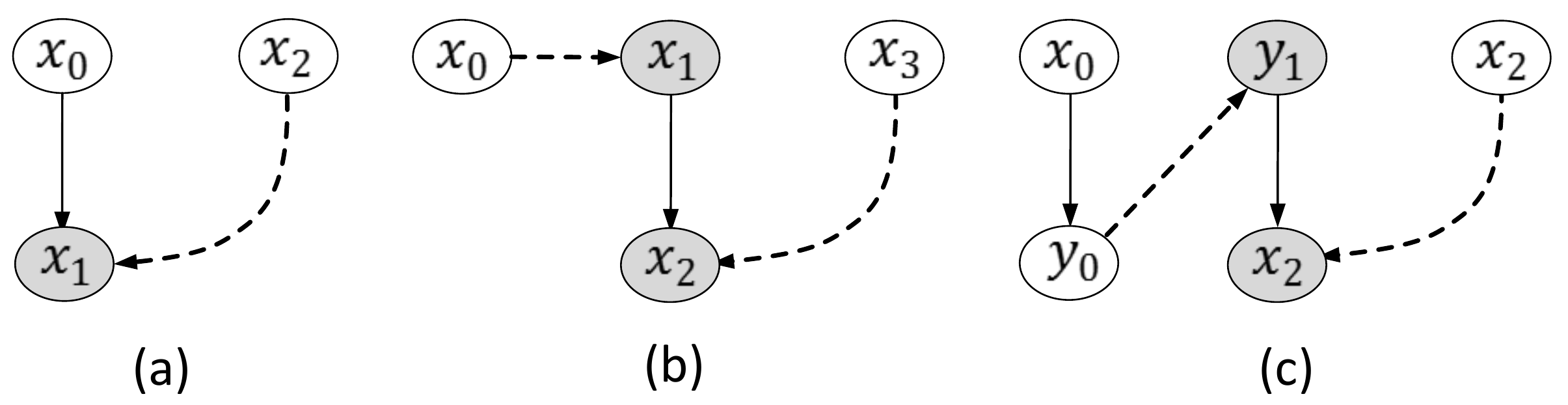}
\caption{Three \eog examples}
\label{fig:6}
\end{figure}

\begin{figure}
\centering
\includegraphics[width=8.5cm]{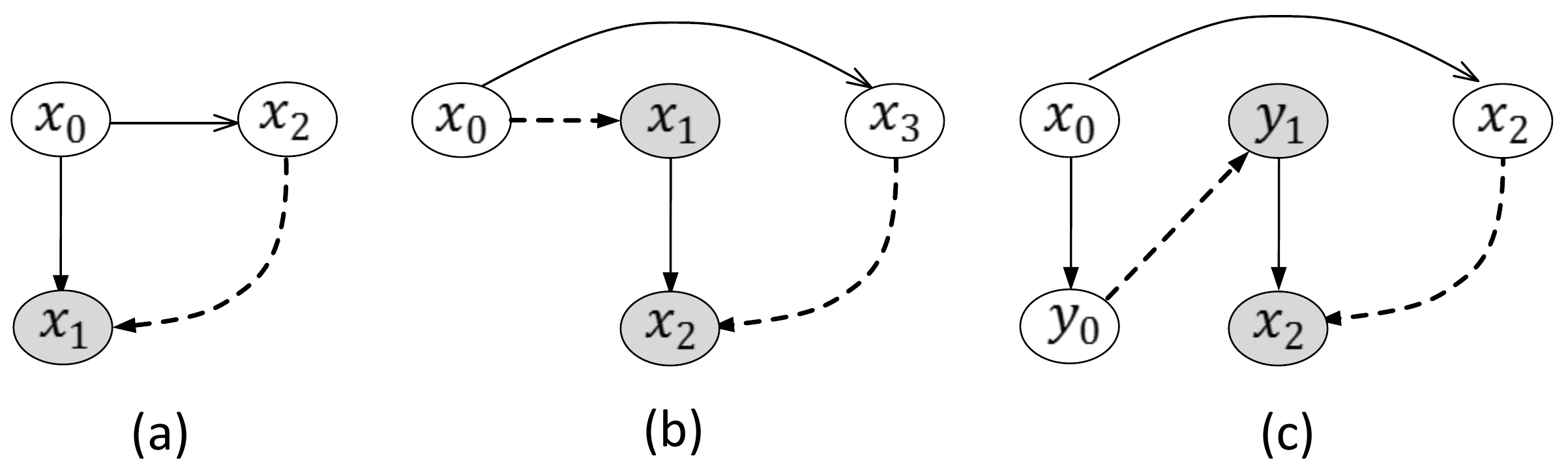}
\caption{Derived orders deduced by applying Rule~\ref{rule:2} to the three {\eog}s shown in Fig.~\ref{fig:6}.}
\label{fig:7}
\end{figure}

Similarly, we propose the following deductive rule:

\begin{rrule}
\label{rule:3}
$\ifrule{e_1\cerfprec e_2,\  e_1\ceprec e_3, \
  \type(e_3)=\ww,\  \var(e_3)=\var(e_1)}
{e_2\ceprec e_3}$.
\end{rrule}

According to these three rules, for the three {\eog}s presented in Fig.~\ref{fig:6}, we can deduce $e_{x_0} \prec_\ce e_{x_2}$, $e_{x_0} \prec_\ce e_{x_3}$, and $e_{x_0} \prec_\ce e_{x_2}$, respectively.
We add these terms to the corresponding {\eog}s as shown in Fig.~\ref{fig:7}.

When a derived order $(e_1, e_2)$ is added to $\ceprec$, some new orders may be propagated via Rules~\ref{rule:1} to \ref{rule:3}.
We repeat this process until we reach a \emph{fixpoint}, i.e., no derived order can be deduced any more.
If some cycle exists in $\ceprec$, then the \eog is infeasible.

Now a conjecture is that, an \eog is feasible if it contains no cycle.
Such a conjecture holds for almost all examples in our experiments.
However, this conjecture may still be false for some special examples. Consider the ``butterfly'' example in Fig.~\ref{fig:8} that involves six threads $\{P_0, P_1, \cdots, P_5\}$ and five shared variables $\{m, n, x, y, w\}$.
No derived order can be deduced, and no cycle exists in the \eog.
However, no total order of $\events_\ce$ can satisfy all the order requirements defined in the \eog, and the \eog is infeasible.

\begin{figure}
\centering
\includegraphics[width=6cm]{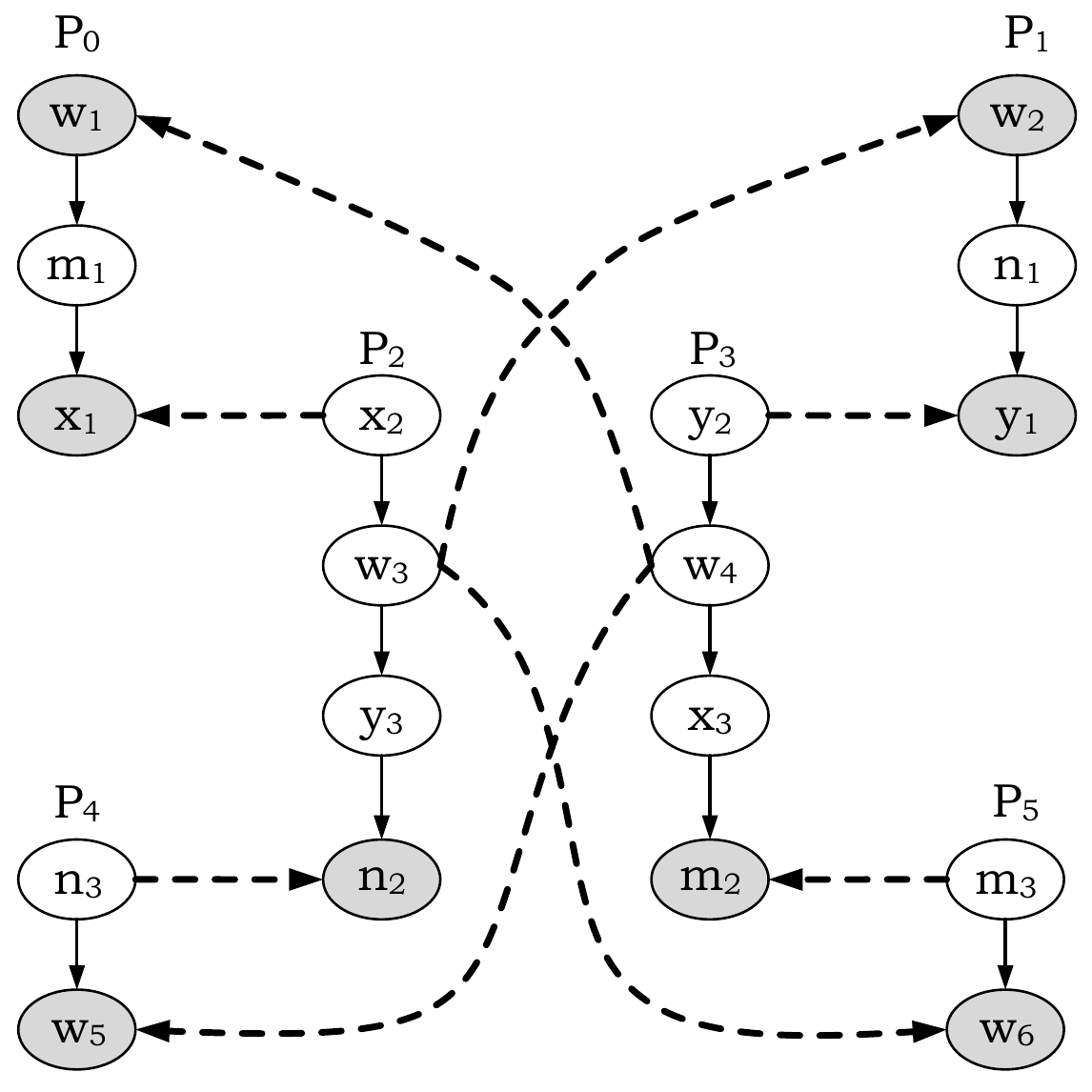}
\caption{The ``butterfly'' example}
\label{fig:8}
\end{figure}

Our graph-based \eog validation method is shown in Algorithm~\ref{alg:eval}.
It repeatedly applies Rules~\ref{rule:1} to~\ref{rule:3} to deduce new derived orders, until a fixpoint is reached.
If there exists some conflict event $e$ s.t. $e\ceprec e$, then the \eog must be infeasible.
Otherwise, it is not sure whether the \eog is feasible or not.

\begin{algorithm}
  \caption{Graph-based \eog validation}
  \label{alg:eval}
  \KwIn{An \eog $G_\ce=\langle\events_\ce,\cepoprec, \cerfprec\rangle$}
  \KwOut{Infeasible: $G_\ce$ is infeasible; Not-Sure: not sure whether $G_\ce$ is feasible or not}
  \Repeat {\rm $\ceprec$ reaches a fixpoint}
  {
    \If{\rm there exist $(e_1,e_2),(e_2,e_3)\in\ceprec$ and $(e_1,e_3)\not\in\ceprec$ }
    {$\ceprec=\ceprec\cup(e_1,e_3)$\;}
    \If{\rm there exists $(e_1,e_2)\in\cerfprec$ exists $(e_3,e_2)\in\ceprec$ and
    $(e_3,e_1)\not\in\ceprec$}
  {\If{\rm $\type(e_3)=\ww$ and $\var(e_3)=\var(e_1)$}{
        $\ceprec:=\ceprec\cup(e_3,e_1)$\;
    }}
    \If{\rm there exists $(e_1,e_2)\in\cerfprec$ exists $(e_1,e_3)\in\ceprec$ and
      $(e_2,e_3)\not\in\ceprec$}
    {\If{\rm $\type(e_3)=\ww$ and $\var(e_3)=\var(e_1)$}
      {$\ceprec:=\ceprec\cup(e_2,e_3)$\;}}
}
\If{\rm $e \ceprec e$ for some $e$}{\Return Infeasible\;}

\Return{Not-Sure}\;
\end{algorithm}

We now prove the correctness of this algorithm.

\begin{theorem}
\label{thm:impliclit_alg1}
If Algorithm~\ref{alg:eval} concludes that $G_\ce$ is infeasible, then $G_\ce$ must be infeasible.
\end{theorem}
\begin{proof}
If Algorithm~\ref{alg:eval} concludes that $G_\ce$ is infeasible, then there must exist some conflict event $e$ s.t. $e \ceprec e$.
Suppose that $G_\ce$ is feasible.
Then according to Definition~\ref{def:eog-feasible}, there must exist a total order $<_\ce$ s.t. $<_\ce$ obeys all order requirements defined in $G_\ce$.
According to Rules~\ref{rule:1} to~\ref{rule:3}, we have $<_\ce$ also obeys all the derived orders deduced in Algorithm~\ref{alg:eval}.
Hence, $\clk(e) < clk(e)$ in $<_\ce$, which is impossible for a total order.

Therefore, the theorem is proved.
\end{proof}

\subsection{Constraint-Based \eog Validation}
If the graph-based \eog validation method is not sure whether an \eog is feasible or not, we employ a constraint solver to further determine feasibility of the \eog.
The method is that, we exactly encode all the order requirements defined in Definition~\ref{def:eog} into a constraint.
The \eog is feasible iff the constraint is satisfiable.
If a SAT is returned, then it generates a total order of all the events which satisfies all the order requirements of the \eog as a byproduct.
Otherwise, both the \eog and the counterexample are infeasible and an unsatisfiable core is generated as a byproduct.

\section{Kernel Reason Based Refinement Generation}
\label{sec:refinement}

If a counterexample is determined to be infeasible, one should add some constraints to the abstraction to prevent this counterexample from appearing again in the future search.
The most intuitive way to address this problem is to add the negation of this counterexample to the abstraction.
However, such a manner excludes only one counterexample in each refinement.
To prune more search space, a better idea is to analyze the kernel reasons that make the counterexample infeasible, and then adds the negation of these kernel reasons to the next abstraction.
Given that a counterexample is validated via validating its corresponding \eog, we analyze the kernel reasons that make an \eog infeasible, s.t. a large amount of space can be pruned in each refinement iteration.

\subsection{Representation of Kernel Reasons}
Given a counterexample $\ce$, the corresponding \eog $G_\ce :=\langle \events_\ce, \cepoprec, \cerfprec\rangle$ is determined by the values of the following two kinds of literals:
1) The values of those \emph{guard condition literals} for the events in $\events$. They determine both $\events_\ce$ and $\cepoprec$. $\events_\ce$ is the set of events appeared in $\ce$. $\cepoprec$ is the restriction of $\ppprec$ to $\events_\ce$.
2) The values of those \emph{read-write link literals} which define the read-from relationship $\cerfprec$.

Let $\mathbb{G}_\ce$ and $\mathbb{S}_\ce$ denote the sets of true guard condition literals and true read-write link literals in $\ce$, respectively.
If $G_\ce$ is infeasible, then the infeasibility could be deduced by the conjunction of all literals in $\mathbb{G}_\ce \cup \mathbb{S}_\ce$.
The infeasibility can often be deduced by a subset of $\mathbb{G}_\ce \cup \mathbb{S}_\ce$.
Therefore, the \emph{kernel reasons} that make a counterexample $\ce$ infeasible could be represented by the minimal subsets of $\mathbb{G}_\ce \cup \mathbb{S}_\ce$ that could deduce the infeasibility.
Finding the minimal subsets not only reduces the constraint size but also prunes more search space.

\subsection{Constraint-Based Kernel Reason Analysis}
If the infeasibility is identified by the constraint-based \eog validation process, the constraint solver can return an unsatisfiable core as a byproduct.
Modern constraint solvers, such as \minisat, allow their users to take a set of literals as assumptions.
When an UNSAT is returned, the constraint solver generates a subset of the assumption literals as an unsatisfiable core.
Based on this idea, we take all literals in $\mathbb{G}_\ce \cup \mathbb{S}_\ce$ as assumption literals.
Whenever the \eog is determined to be infeasible, the constraint solver generates a subset of $\mathbb{G}_\ce \cup \mathbb{S}_\ce$ that can still deduce the infeasibility.
Suppose that it is $\{\ell_1, \ell_2, \cdots, \ell_n\}$.
Let $\krf := \ell_1 \wedge \ell_2 \wedge \cdots \wedge \ell_n$.
The refinement constraint can then be formulated as follows.
\begin{eqnarray}
\label{eq:usc}
\kappa := \neg \krf
\end{eqnarray}

If the constraint that exactly encodes all order requirements of an \eog is unsatisfiable, it may have a large number of unsatisfiable cores.
To prune as much search space as possible, one should obtain all of them.
However, generating all unsatisfiable cores is usually time consuming.
Most constraint solvers generate only one unsatisfiable core each time, and it may not be the shortest one, which significantly limits the pruned search space in each refinement.
Hence, constraint solving is not a good choice for our refinement generation.

\subsection{Graph-Based Kernel Reason Analysis}
This section presents our graph-based kernel reason analysis method.
Compared with the constraint-based method, it usually obtains a much more effective refinement.
It can efficiently obtain all kernel reasons that make an \eog infeasible.
In addition, the obtained ``core kernel reasons'' are always the shortest ones.
Hence, it can usually prune much more search space in each iteration.

In Algorithm~\ref{alg:eval}, if an infeasibility is determined, then there must exist some conflict event $e$ s.t. $e \prec_\ce e$.
A conflict event can usually be attributed to several kernel reasons, and there are usually many conflict events.
To prune more search space, one should find all kernel reasons of all conflict events.

We define ``a kernel reason of an order $\order \in \ceprec$'' to be ``the minimal subset of $\mathbb{G}_\ce \cup \mathbb{S}_\ce$ that can deduce $\order$''.
When a derived order $\order$ is deduced, if the kernel reasons of all existing orders (including existing derived orders) are given, then we can obtain a set of kernel reasons of $\order$ upon its production.
Note that a derived order $\order$ may be deduced for multiple times.
Whenever it is deduced, we can obtain new kernel reasons of $\order$.
Based on this observation, we initialize the kernel reasons of every order $\order$ to $\emptyset$.
The kernel reasons of an order $\order$ are then updated once $\order$ is added (we add the orders in $\cepoprec$ and $\cerfprec$ into the graph one by one) or deduced.
In this manner, we obtain the kernel reasons of each order $\order \in \ceprec$ when Algorithm~\ref{alg:eval} terminates.

We then discuss how to update the kernel reasons of an order $\order \in \ceprec$ when $\order$ is added or deduced.
We hypothesize that the kernel reasons of all existing orders are given.
We denote by $\kr(\order)$ and $\krset(\order)$ a kernel reason and the set of kernel reasons of $\order$, respectively.

\begin{itemize}
  \item If $\order := (e_1, e_2)$ is added because $\order \in \cepoprec$, then $e_1 \ceprec e_2$ iff both $e_1, e_2 \in \events_\ce$, i.e., the guard condition literals for both $e_1$ and $e_2$ are true.
    Therefore, $\krset(\order) := \krset(\order) \cup \{\{\guard(e_1), \guard(e_2)\}\}$, where $\guard(e)$ is the guard condition literal of $e$.

  \item If $\order := (e_1, e_2)$ is added because $\order \in \cerfprec$, $\order$ holds iff $e_2$ reads the value written by $e_1$, which already indicates that both $\guard(e_1)$ and $\guard(e_2)$ are true.
      Therefore, $\krset(\order) := \krset(\order) \cup \{\{\sel(\lambda)\}\}$, where $\sel(\lambda)$ is the read-write link literal of $\order$.
  \item If $\order$ is a derived order deduced from $\order_1$ and $\order_2$, then $\order \in \ceprec$ iff both $\order_1$ and $\order_2$ belong to $\ceprec$.
      Therefore, $\krset(\order):=\krset(\order) \cup \{\kr_1\cup\kr_2 \mid \kr_i\in\krset(\lambda_i)\}$. Note that $\order$ may be deduced for multiple times. We incrementally update $\krset(\order)$ whenever $\order$ is deduced.
\end{itemize}

A kernel reason $\kr \in \krset(\order)$ is considered \emph{redundant} if some kernel reason $\kr' \in \krset(\order)$ exists s.t. $\kr' \subseteq \kr$.
Such kernel reasons are immediately eliminated from $\krset(\order)$, and the remaining reasons are called \emph{core kernel reasons} of $\order$.
Using this strategy, we maintain only the set of core kernel reasons in our algorithm.
This set is dynamically updated during the running of the algorithm.

In this manner, we obtain the set of core kernel reasons of any order $\order \in \ceprec$ when Algorithm~\ref{alg:eval} terminates.
Let $\conflicts$ denote the set of conflict events.
The kernel reasons that make $\ce$ infeasible (denoted by $\krset(\ce)$) can be expressed as follows.

\begin{equation}\label{eq:kr}
  \krset(\ce) := \bigcup_{e \in \conflicts}\krset(e \ceprec e)
\end{equation}

Again, we eliminate those redundant kernel reasons from $\krset(\ce)$, and maintain only those core kernel reasons.
In our experiments, hundreds or even thousands of kernel reasons may be observed, but only several or dozens of them are considered core ones.

Suppose that $\kr := \{\ell_1, \ell_2, \cdots, \ell_m\}$ where each $\ell_i$ is a literal,
we define the following:
\begin{equation}
\label{equ:krc}
  \krf_{\kr} := \ell_1 \wedge \ell_2 \wedge \cdots \wedge \ell_m
\end{equation}

Suppose that $\krset(\ce) := \{\kr_1, \kr_2, \cdots, \kr_n\}$, then the refinement constraint is formulated as follows.
\begin{equation}
\label{equ:refinement}
  \kappa := \bigwedge_{i=1}^n\neg \krf_{\kr_i}
\end{equation}

Algorithm~\ref{alg:refienment} demonstrates our graph-based refinement generation method.
We first add all the program orders into $\ceprec$, and update their kernel reasons according to the kernel reason updating method we have discussed.
Adding a read-from order or a derived order to $\ceprec$ may propagate a large number of new orders.
Whenever a read-from order $\order \in \cerfprec$ is added to $\ceprec$, we denote by $\mathbb{D}$ the set of orders that will be added to $\ceprec$ before another read-from order is added.
Adding each order $\order' \in \mathbb{D}$ to $\ceprec$ may propagate a set of derived orders, which are denoted by $\mathbb{B}$.
Whenever an order $\order'' \in \mathbb{B}$ is deduced, we update its kernel reasons and add it to $\mathbb{D}$ if it is not contained in $\ceprec$.
In this manner, once all read-from orders have been added to $\ceprec$, we obtain all derived orders and the kernel reasons of all orders in $\ceprec$.
We then compute the refinement constraint according to equation (\ref{eq:kr}), (\ref{equ:krc}) and (\ref{equ:refinement}).

\begin{algorithm}
  \caption{Graph-based refinement generation.}
  \label{alg:refienment}
  \KwIn{An \eog $G_\ce :=\langle\events_\ce,\cepoprec, \cerfprec\rangle$.}
  \KwOut{A refinement constraint $\kappa$.}
    $\ceprec := \cepoprec$, and update $\krset(\order)$ for each $\order \in \cepoprec$\;
    \ForEach{$\order \in \cerfprec$}{
        Update $\krset(\order)$, and let $\mathbb{D} := \{\order\}$\;
        \ForEach{$\order' \in \mathbb{D}$}{
            $\ceprec := \ceprec \cup \{\order'\}$\;
            Let $\mathbb{B}$ be the set of propagated orders due to $\order'$\;
            \ForEach{$\order'' \in \mathbb{B}$}{
                Update $\krset(\order'')$\;
                \If{$\order'' \notin \ceprec$}{
                    $\mathbb{D} := \mathbb{D} \cup \{\order''\}$\;
                }
            }
        }
    }
    Compute $\kappa$ according to equation (\ref{eq:kr}), (\ref{equ:krc}) and (\ref{equ:refinement})\;
    \Return $\kappa$\;
\end{algorithm}

From the above discussion, the graph-based refinement generation method has the following advantages:
\begin{enumerate}
\item It can detect all kernel reasons that make $\ce$ infeasible, leading to a large amount of search space pruned in each iteration.
\item It maintains a minimal subset of core kernel reasons, thereby making the refinement constraint small and manageable.
\end{enumerate}

\subsection{Correctness of the Graph-Based Kernel Reason Analysis}

We prove the correctness of our graph-based refinement generation method, that is, the refinement constraint obtained in equation (\ref{equ:refinement}) should be true in any feasible counterexample, i.e., it will not eliminate any feasible counterexample from the abstraction.

\begin{theorem}
\label{thm:krinsensitive}
Given a counterexample $\ce$ and a kernel reason $\kr(\order)$ obtained according to our graph-based kernel reason analysis method, we have $\kr(\order) \models \order$. That is, for any other counterexample $\ce'$, we also have $\order\in\prec_{\ce'}$ if $\ce'\models\krf_{\kr(\order)}$.
\end{theorem}

\begin{proof}
We prove this theorem by induction.

\textbf{Inductive Base}:
If $\kr(\order)$ is obtained because $\order \in \cepoprec$ or $\order \in \cerfprec$, then the conclusion can be immediately inferred from the definition.

\textbf{Inductive Step}:
If $\kr(\order)$ is obtained via Rule~\ref{rule:1}, \ref{rule:2}, or \ref{rule:3}, then two orders $\order_1$ and $\order_2$ must exist, such that $\kr(\order) = \kr(\order_1) \cup \kr(\order_2)$. Given that $\krf_{\kr(\order)}$ holds w.r.t. $\ce'$, $\krf_{\kr(\order_1)}$ and $\krf_{\kr(\order_2)}$ are also true w.r.t. $\ce'$. By applying the induction hypothesis, we obtain $\order_1, \order_2 \in \prec_{\ce'}$.
According to the same rule, we deduce that $\order \in \prec_{\ce'}$.
\end{proof}

\begin{theorem}
\label{thm:krcorr}
Given a kernel reason $\kr \in \krset(\ce)$ of an infeasible counterexample $\ce$, for any feasible counterexample $\ce'$, we have $\ce'\models\neg\krf_\kr$.
\end{theorem}

\begin{proof}
Given that $\kr$ is a kernel reason that makes $\ce$ infeasible, according to equation~(\ref{eq:kr}), there must exist a corresponding order $\order := (e, e) \in \ceprec$.
Suppose that $\ce'\models\krf_\kr$.
Then according to Theorem~\ref{thm:krinsensitive}, we have $\order := (e, e) \in \prec_{\ce'}$.
It indicates that $\ce'$ is infeasible, which is contradict with that $\ce'$ is feasible.
Hence, we must have $\ce'\models\neg\krf_\kr$.
\end{proof}

\begin{theorem}
\label{thm:rccorr}
Given a refinement constraint $\kappa$ obtained in some iteration, for any feasible counterexample $\ce'$, we have $\ce'\models\kappa$.
\end{theorem}

\begin{proof}
Suppose that $\kappa := \bigwedge_{i=1}^n\neg \krf_{\kr_i}$, and the corresponding counterexample is $\ce$.
According to Theorem~\ref{thm:krcorr}, for any $i$ ($1 \leq i \leq n$), we have $\ce'\models\neg\krf_{\kr_i}$.
Hence, $\ce'\models\kappa$.
\end{proof}

\section{Soundness and Efficiency}
\label{sec:soundness}

\subsection{Soundness and Completeness}
We prove the soundness and completeness of our method (shown in Fig.~\ref{fig:x}) via three theorems.

\begin{theorem}
\label{thm:true}
If our method concludes that the property is safe, then it must be safe w.r.t. the given loop unwinding depth.
\end{theorem}

\begin{proof}
To prove this theorem, we prove that $\monolithic\wedge\phi_{err} \models \varphi_0 \wedge \wedge_{i=0}^n \kappa_i \wedge \phi_{err}$, where $\monolithic$ is the monolithic encoding, $\varphi_0$ is the initial abstraction, $\kappa_i$ is the $i$-th refinement constraint, and $\phi_{err}$ is the error states.
According to the definition of $\monolithic$ and $\varphi_0$ (cf. Definition~\ref{def:abstraction}), we can obtain $\monolithic\wedge\phi_{err} \models \varphi_0 \wedge \phi_{err}$.
We prove that $\monolithic\wedge\phi_{err}\models\kappa_i$ as follows.

If $\kappa_i$ is obtained from the graph-based refinement process, we prove that for each element $\neg \krf \in \kappa_i$, $\monolithic\wedge\phi_{err}\models\neg \krf$.
Suppose that $\monolithic\wedge\phi_{err}\nvDash\neg \krf$.
Then there must exist an assignment $\pi$ of $\monolithic\wedge\phi_{err}$ s.t. $\krf$ holds.
Given that $\pi$ is a feasible counterexample, according to Theorem~\ref{thm:krcorr}, $\pi \models \neg\krf$, which is contradict with that $\krf$ holds in $\pi$.
Hence, we must have $\monolithic\wedge\phi_{err}\models\neg \krf$ and $\monolithic\wedge\phi_{err}\models \kappa_i$.

If $\kappa_i$ is obtained from the constraint-based refinement process, then $\kappa_i = \neg \krf$ where $\krf$ is an unsatisfiable core of the formula which encodes all order requirements of an \eog.
Suppose that $\monolithic\wedge\phi_{err}\nvDash\neg \krf$.
Then there must exist an assignment $\pi$ of $\monolithic\wedge\phi_{err}$ s.t. $\krf$ holds, which is contradict with that $\krf$ is an unsatisfiable core.
Hence, we must have $\monolithic\wedge\phi_{err}\models\kappa_i$.
\end{proof}

\begin{theorem}
\label{thm:false}
If our method concludes that the property is unsafe, then a true counterexample of the property must exist.
\end{theorem}

\begin{proof}
In our method, the property is concluded unsafe only if the constraint-based \eog validation process returns \sat.
Given that the formula in the constraint-based \eog validation process encodes all order requirements of the \eog exactly, according to Definition~\ref{def:eog} and~\ref{def:eog-feasible}, the \eog is feasible iff the formula is satisfiable.
Hence, the \eog is feasible.
According to Theorem~\ref{thm:eog_and_ce}, the corresponding counterexample must be feasible.
Therefore, a true counterexample of the property must exist.
\end{proof}

\begin{theorem}
\label{thm:terminate}
Our method will terminate for any program with finite state space.
\end{theorem}

\begin{proof}
For a multi-threaded program with finite state space, the number of counterexamples of the initial abstraction must be finite.
Suppose that the counterexample obtained in the $i$-th iteration is $\ce_i$.
According to Section~\ref{sec:refinement}, each kernel reason of the counterexample or the unsatisfiable core obtained in the constraint-based refinement process is just a subset of $\mathbb{G}_{\ce_i} \cup \mathbb{S}_{\ce_i}$.
According to equation (\ref{equ:refinement}) and (\ref{eq:usc}), we have $\ce_i$ must be absent in the next abstraction.
Hence, our method reduces at least one counterexample in each iteration.
It terminates when all counterexamples have been reduced or a true counterexample is found.
\end{proof}

\subsection{Efficiency}

A both efficient and sound way for counterexample validation and refinement generation will be elegant. However, such procedure is usually difficult to devise.
As an alternative, we integrate the graph-analysis and constraint solving approaches together to obtain a both efficient and sound method.

We have proved that enhanced by the constraint-based counterexample validation and refinement generation processes, our method is sound and complete.
We now analyze the effectiveness of the graph-based \eog validation method.
If the \eog is infeasible, the infeasibility may be determined by either the graph-based \eog validation process or the constraint-based \eog validation process.
Although both of these processes can rapidly determine such infeasibility, a much more effective refinement can be obtained if the infeasibility is determined from the graph-based \eog validation process.
Fortunately, cases similar to the ``butterfly'' example rarely occur in practice.
In other words, the graph-based \eog validation process can always identify the infeasibility with rare exceptions.

Suppose that the verification problem is solved via $n$ iterations.
If the property is determined to be unsafe, then all {\eog}s generated during the first $n-1$ rounds must be infeasible,
which can generally be identified by the graph-based \eog validation process.
The constraint-based \eog validation process is only invoked in the last iteration during which the property is violated.
If the property is proved safe w.r.t. the given loop unwinding depth, then the infeasibility of all the $n$ infeasible {\eog}s can generally be identified by the graph-based \eog validation process, and the constraint-based \eog validation process will not be invoked.

In sum, advantages of our method include: 1) Without the \tsc, the initial abstraction $\varphi_0$ is usually much smaller than the monolithic encoding. 2) The graph-based refinement process can usually obtain a small yet effective refinement, which reduces a large amount of space in each iteration whereas the size of all those refinement constraints can usually be ignored compared with that of the abstraction. 3) Though the graph-based validation process is not complete, it is effective to identify the infeasibility in practice.

\section{Experimental Results}
\label{sec:experiment}

We have implemented our method on top of CBMC-4.9 \footnote{Downloaded from \url{https://github.com/diffblue/cbmc/releases} on Nov 20, 2015} and employed \textsc{MiniSat2} as the back-end constraint solver.
Our tool is named \gcbmc, and it is available at \cite{sv-comp2017}.
We use the 1047 multi-threaded programs of \svcompc \cite{sv-comp2017} as our benchmarks.
In the experiments, our tool supports nearly all features of C language and PThreads.

\subsection{Benchmark of \svcompc}
The open-source, representative, and reproducible benchmarks of Competition on Software Verification (\svcomp) have been widely accepted for program verification.
Given that these benchmarks are devised for comparison of those state-of-the-art techniques and tools, a significant number of studies on concurrent program verification have performed their experiments on them.

The concurrency benchmarks of \svcompc include 1047 examples and cover most of the publicly available concurrent C programs that are used for verification.
Though many of these examples are small in size in previous years, dozens of complex examples have been added to these benchmarks in recent years.
For instance, the examples in the \texttt{pthread-complex} directory (collected by the CSeq team) are taken from the papers on PLDI¡¯15 \cite{MachadoLR15}, POPL¡¯15 \cite{BouajjaniEEH15}, and PPOPP¡¯14 \cite{ThomsonDB14}, which are used for concurrent program debugging and testing;
the examples in the \texttt{pthread-driver-races} directory (collected by the SMACK+CORRAL team) are used for symbolic analysis of the drivers from the Linux 4.0 kernel \cite{DeligiannisDR15};
and the examples in the \texttt{pthread-C-DAC} directory (collected by C-DAC) are from the industrial problems of Centre for Development of Advanced Computing, Pune, India.
These programs contain hundreds of lines, 4 to 8 threads, complex structure variables with 2D pointers, and hundreds or even a thousand read/write accesses \footnote{A read/write access of a complex structure variable may contain hundreds of read/write accesses of boolean variables. Here a read/write access of a complex structure variable is considered just one read/write access.}.
Given these complex features, these programs are challenging for existing state-of-the-art concurrency verification techniques and tools.

\subsection{Experimental Setup}
We conduct all of our experiments using a computer with Intel(R) Core(TM) i5-4210M CPU 2.60 GHz and 12 GB memory.
A 900-second time limit is observed.

We select and compare two classes of the state-of-the-art tools with our method.
The first class is those top winners in recent competitions, including \mucseq \footnote{Downloaded from \url{http://sv-comp.sosy-lab.org/2017/systems.php} on January 24, 2017} \cite{TomascoNI0TP16} and \lazycseq \footnotemark[\value{footnote}] \cite{InversoN0TP17}.
\mucseq is the gold medal winner of \svcomp 2016, and \lazycseq is the silver medal winner of \svcomp 2017.
The second class comprises those tools which methods are closely related to ours, including \cbmc \cite{AlglaveKT13} and \threader \footnote{Downloaded from \url{http://sv-comp.sosy-lab.org/2014/participants.php} on January 24, 2017} \cite{PopeeaR13}.
\cbmc is a highly popular verifier for program verification.
Different from our method, it provides an exact encoding of the \tsc for multi-threaded program verification.
For \threader, to the best of our knowledge, it is the best CEGAR-based verifier for multi-threaded C programs.
It has received the gold medal of the concurrency track of SV-COMP 2013.

Given that different tools employ different techniques, each of which has its own features, it is difficult to make the comparison absolutely fair.
For example, \threader performs unbounded verification, while all the other tools perform BMC; and given a loop unwinding depth, both \mucseq and \lazycseq are incomplete whereas all the other tools are complete.
To make the comparison as fair as possible, we select the latest available version of them, and set the parameters of them to be that of the competition.
We believe that these tools should perform best under these parameters.
For the loop unwinding depth, we set that of \cbmc and \gcbmc the same as that of \mucseq.
Specifically, it is dynamically determined through syntax analysis.
The bound is set to 2 for programs with arrays, and $n$ if some of the program's \texttt{for}-loops are upper bounded by a constant $n$ \cite{TomascoNI0TP16}.
Given that our method is implemented on top of \cbmc.
The only difference between \cbmc and \gcbmc is that \cbmc employs the monolithic encoding while we perform our abstraction refinement on the \tsc.

\subsection{Effectiveness and Efficiency}

\gcbmc solved\footnote{\emph{Solve} means that the verifier gives a correct answer (true/false) within the time limit. Refer to \cite{sv-comp2017} for the rules of the competition.} all the 1047 concurrent programs and has received the highest score of 1293 points.
It has won the gold medal in the concurrency track of \svcompc \cite{sv-comp2017} (Warning: It will violate our anonymity).

\paragraph{Overall comparison with state-of-the-art tools.}
Fig.~\ref{fig:9} compares our tool with the state-of-the-art tools, including \mucseq, \lazycseq, and \cbmc.
Similar to \svcompc, we perform our experiments on BenchExec \footnote{\url{https://github.com/sosy-lab/benchexec}} to achieve a reliable and repeatable benchmarking.

The experimental results for \lazycseq and \mucseq are consistent with those in \svcompc.
However, the results in our experiments for \cbmc is better than those in \svcompc.
The reason is that we have improved \cbmc in several aspects, and we have also realized some concurrency-related improvements in \cbmc-5.5 \footnote{Released on August 20, 2016 in \url{https://github.com/diffblue/cbmc/releases}}.

Fig.~\ref{fig:subfig:9a} compares the overall performance of \lazycseq, \mucseq, \cbmc, and \gcbmc based on the SV-COMP rules.
The $x$-axis represents the accumulated score, while the $y$-axis represents the time needed to achieve a certain score~\footnote{The rules to assign the score can be found in \url{http://sv-comp.sosy-lab.org/2017/rules.php}}.
Both our tool and \lazycseq have successfully solved all examples and obtained 1293 points, while \mucseq and \cbmc obtained 1243 and 1258 points, respectively.
\lazycseq, \mucseq, and \cbmc spent 9820, 2540, and 12300 s to finish all examples, respectively, while our tool completed all examples within 1550 s.

\begin{figure}
	\centering
	\subfigure[Time comparison]{
		\label{fig:subfig:9a}
		\includegraphics[width=5cm]{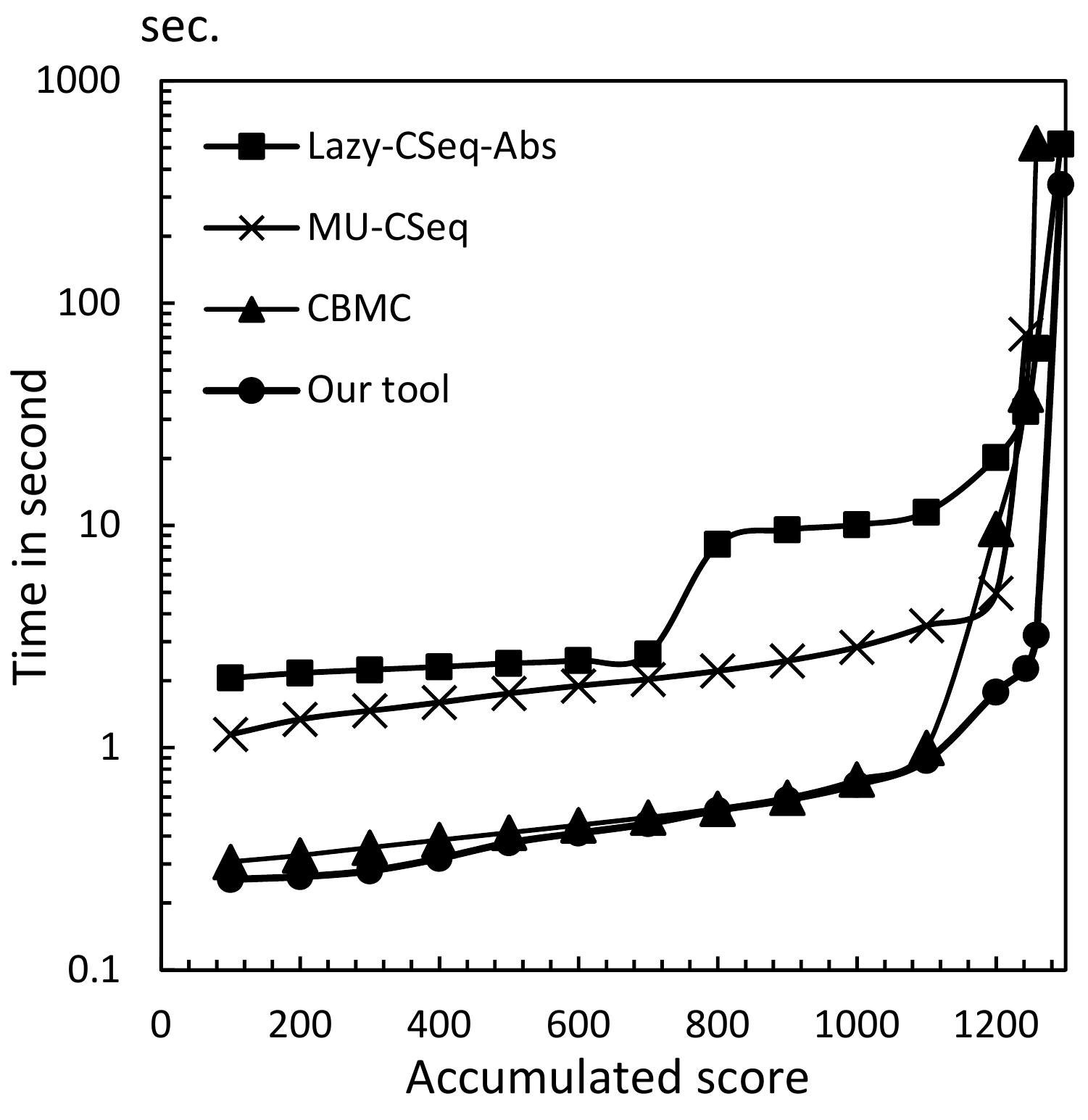}
	}
	\hspace{0.01in}
	\subfigure[Memory comparison]{
		\label{fig:subfig:9b}
		\includegraphics[width=5cm]{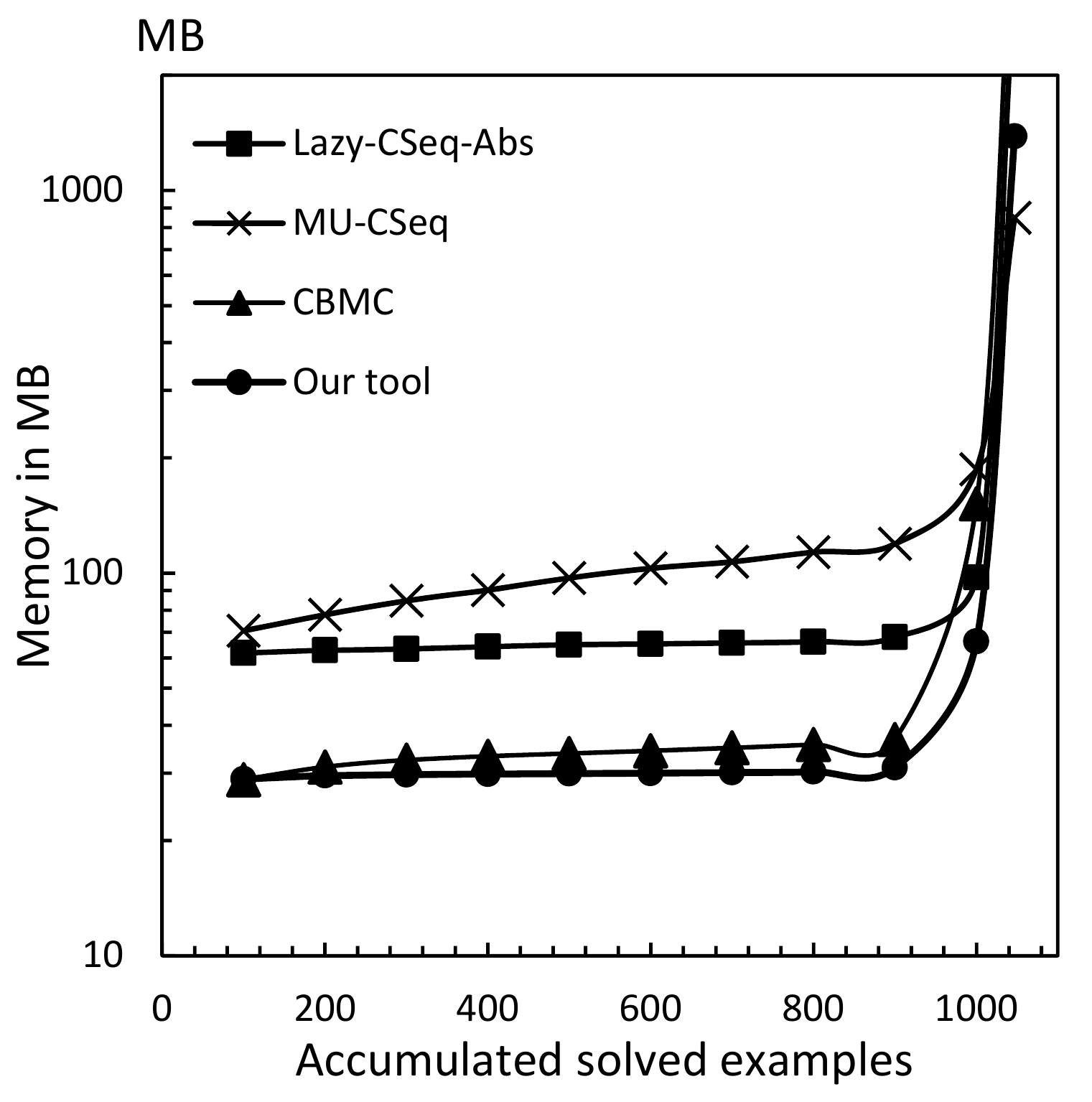}
	}
	\caption{Compare with state-of-the-art tools}
	\label{fig:9}
\end{figure}

Fig.~\ref{fig:subfig:9b} shows the overall memory consumption of the aforementioned tools.
\lazycseq, \mucseq, \cbmc, and \gcbmc require 104, 103, 84 and 43 GB to solve all 1047 examples, respectively.
Given that the scheduling constraint is ignored in the abstraction, our tool always solves small problems and consumes much less memory than the three other tools.

We further compare our tool with \lazycseq, \mucseq, \cbmc, and \threader to evaluate its performance.

\paragraph{\gcbmc versus \lazycseq.}
Compared with \lazycseq, our tool runs 6.34 times faster on average, and consumes only 41\% of the memory over all 1047 examples.
As shown in Fig.~\ref{fig:subfig:10a}, \lazycseq outperforms our tool in only 12 of the 1047 examples.
With its \emph{abstract interpretation} technique, \lazycseq outperforms our tool in those examples where the numerical analysis dominates the complexity.


\paragraph{\gcbmc versus \mucseq.}
\mucseq fails to solve 30 of the 1047 examples.
Compared with this tool, our tool runs 2.43 times faster on average and consumes only 36\% of the memory for the remaining 1017 examples.
As shown in Fig.~\ref{fig:subfig:10b}, \mucseq outperforms our tool in only 33 of the 1047 examples.
By applying the \emph{memory unwinding} technique to limit the number of writes, the encoding size of \mucseq is insensitive to the scale of the data structures, thereby outperforming our tool for some special examples.

\begin{figure}
	\centering
	\subfigure[Our tool versus \lazycseq]{
		\label{fig:subfig:10a} 
		\includegraphics[width=5cm]{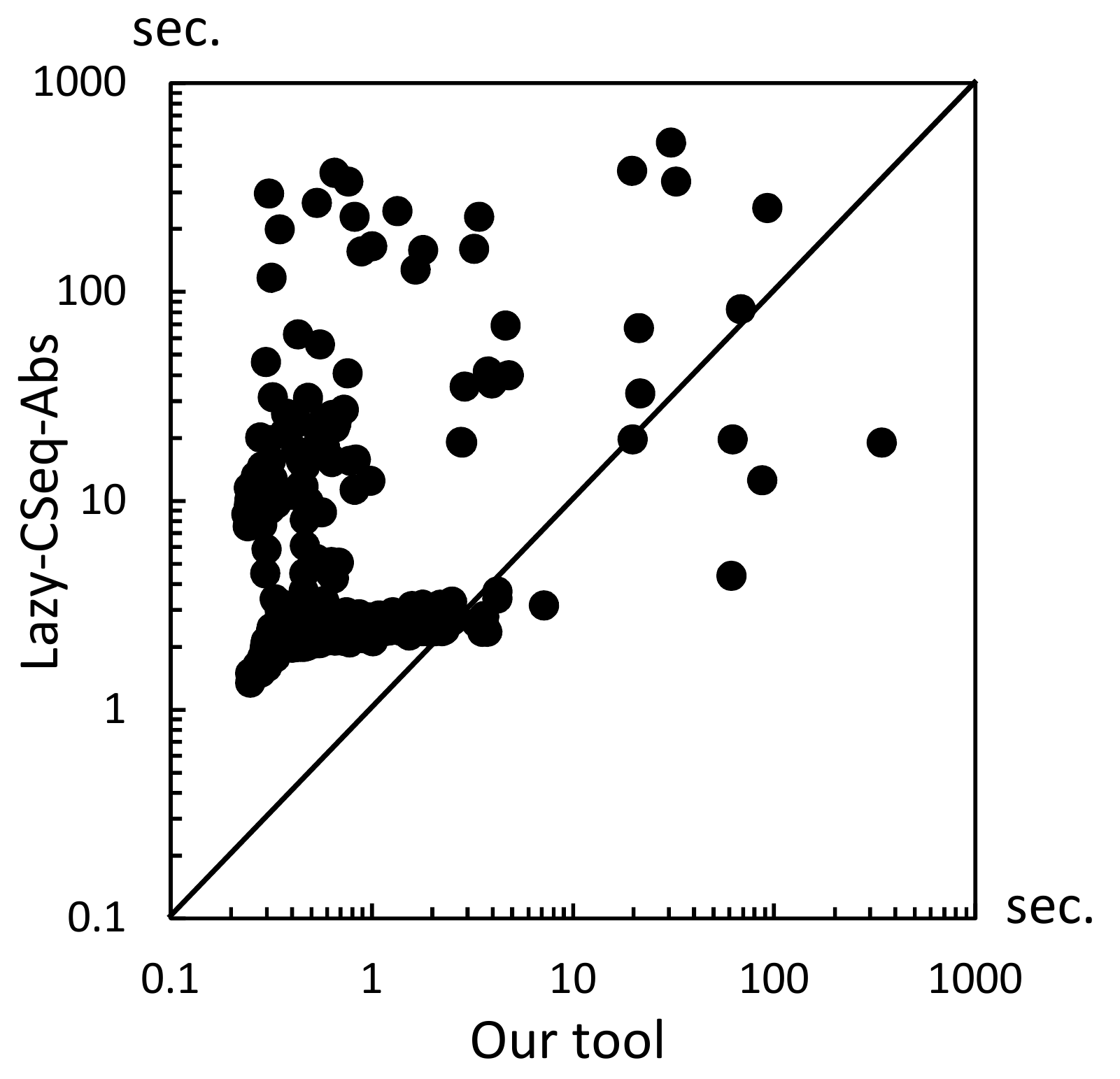}
	}
	\hspace{0.01in}
	\subfigure[Our tool versus \mucseq]{
		\label{fig:subfig:10b} 
		\includegraphics[width=5cm]{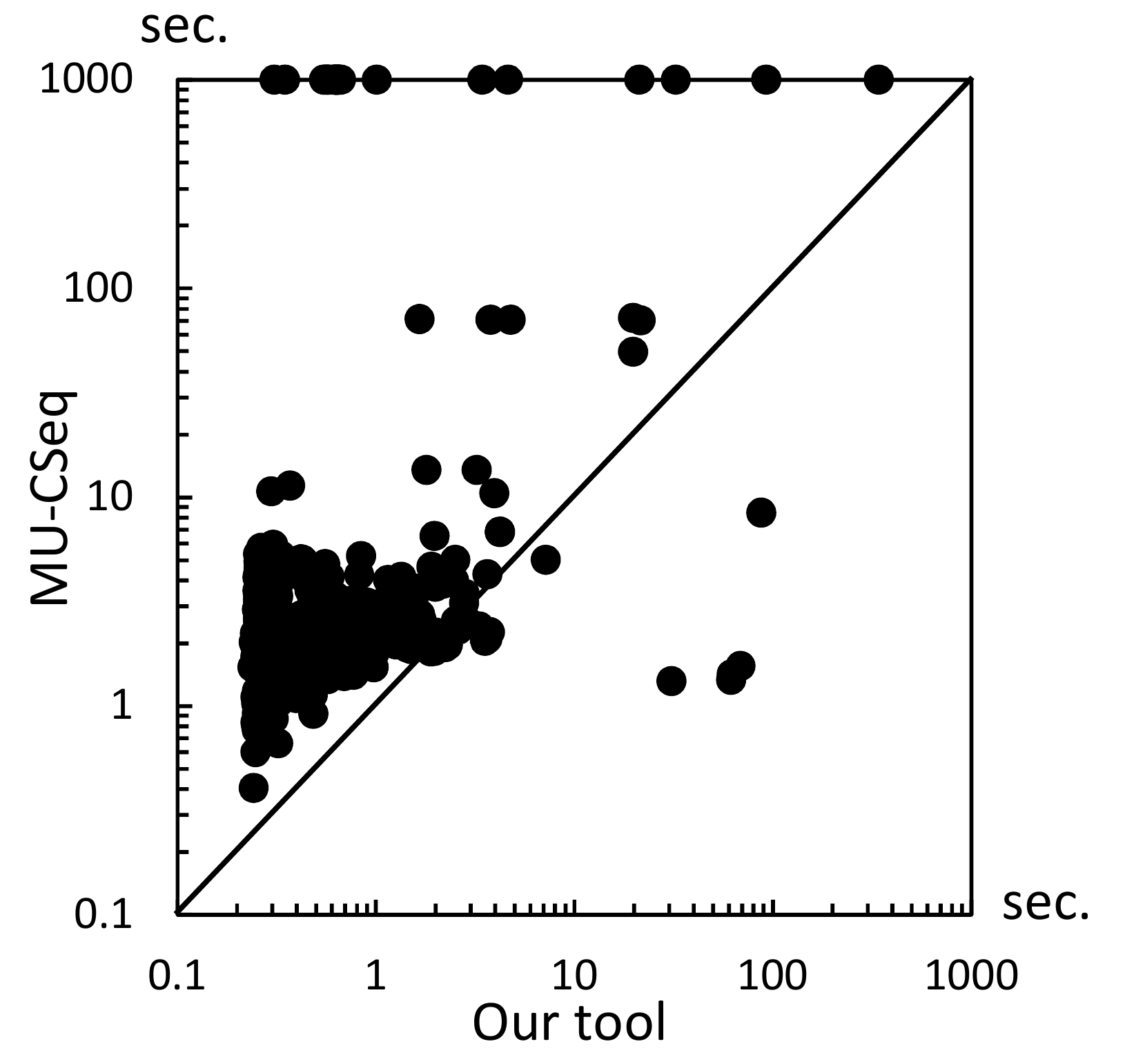}
	}
	\caption{Compare with \lazycseq and \mucseq}
	\label{fig:10} 
\end{figure}

\paragraph{\gcbmc versus \cbmc}
Fig.~\ref{fig:subfig:11a} compares \cbmc with our tool.
\cbmc fails to solve 23 of the 1047 examples.
Both \cbmc and our tool can easily solve 92\% of these examples.
For these trivial examples, the monolithic encoding may sometimes run faster.
However, our tool outperforms \cbmc by 35.8 times on average in the 56 complex cases in which \cbmc needs more than two seconds to solve them.

\paragraph{\gcbmc versus \threader.}
\threader only participated in \svcomp 2013 and 2014.
Given that this tool cannot solve many of the examples in \svcompc, we compare it with our tool on the benchmarks of \svcomp 2014, which contain only 78 examples.
Fig. \ref{fig:subfig:11b} presents the results.
Our tool has completed all 78 examples within the time limit, while \threader completed only 59 examples.
Moreover, our tool and \threader require 140 s and 6865 s to solve these examples respectively, thereby showing that our tool is 49 times faster than \threader on average.
However, as we have declared before, \threader performs unbounded verification, while we perform BMC.

\begin{figure}
	\centering
	\subfigure[Our tool versus \cbmc]{
		\label{fig:subfig:11a} 
		\includegraphics[width=5cm]{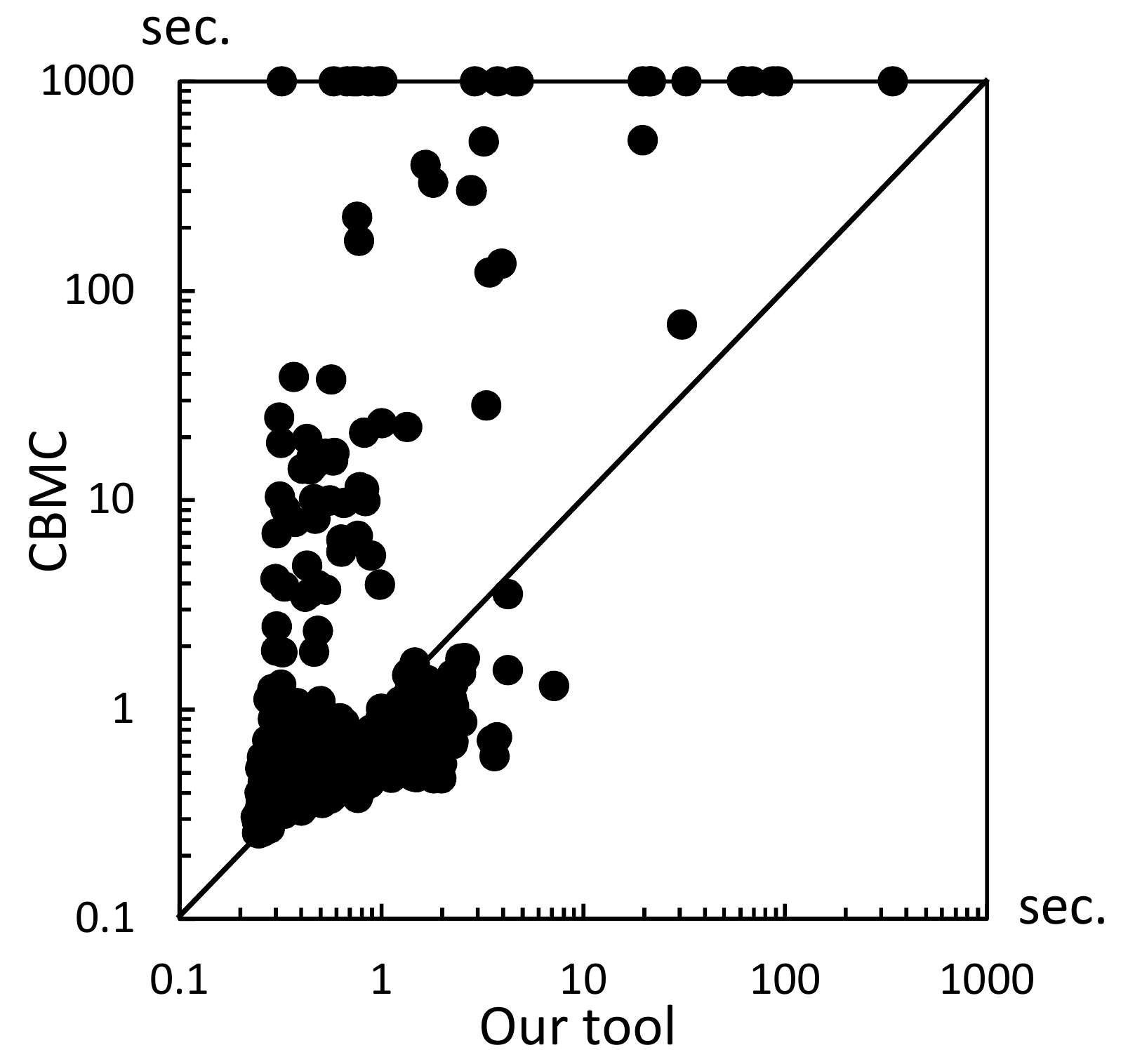}
	}
	\hspace{0.01in}
	\subfigure[Our tool versus \threader]{
		\label{fig:subfig:11b} 
		\includegraphics[width=5cm]{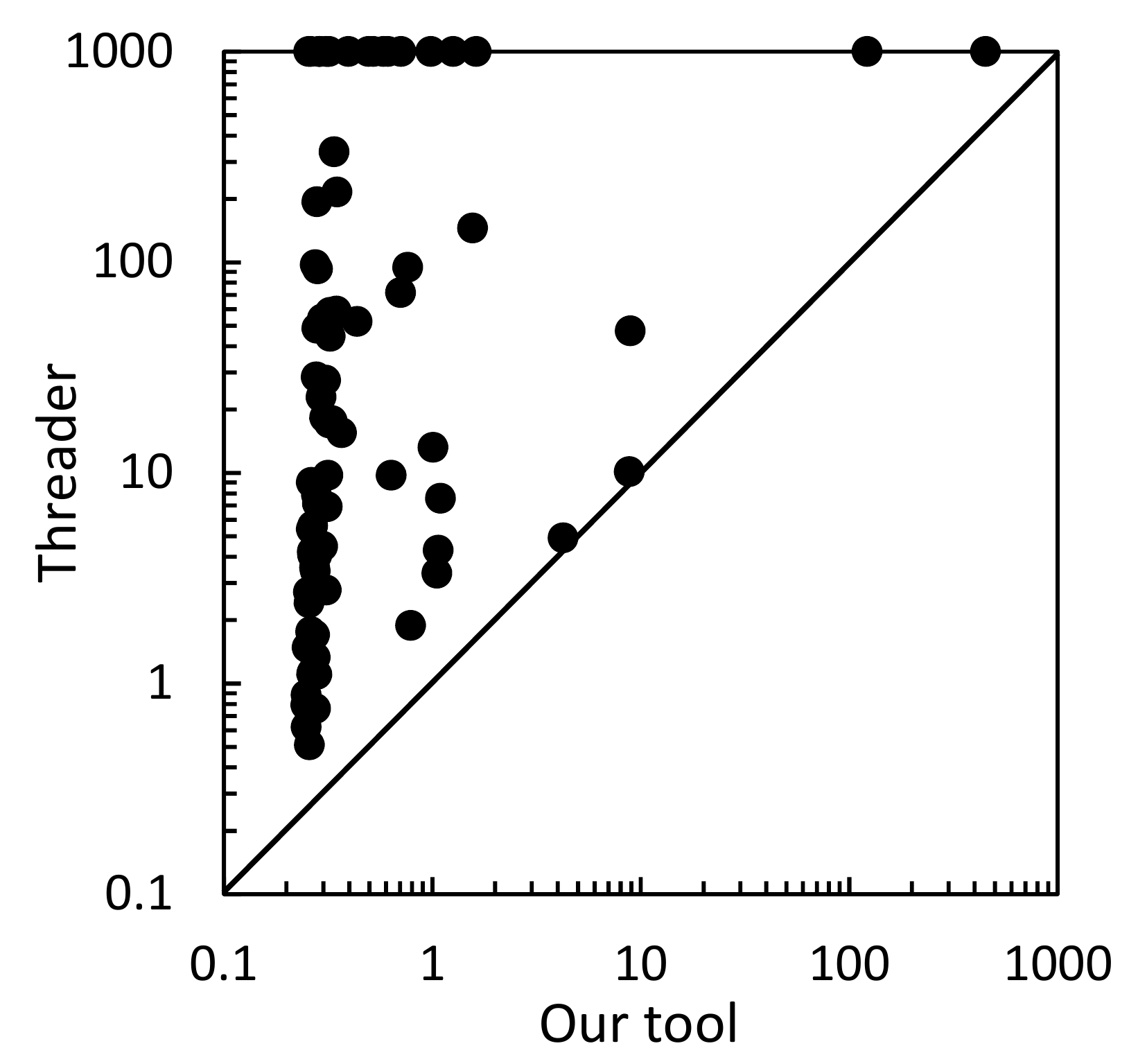}
	}
	\caption{Compare with \cbmc and \threader}
	\label{fig:11} 
\end{figure}

\subsection{Essence Analysis}
\label{costanalysis}

The performance of our tool is mainly affected by the number of refinements, size of constraints and cost of constraint solving, etc. We also justify the benefits from the graph-based refinement generation method.

\paragraph{Number of refinements.}
Fig.~\ref{fig:subfig:14a} presents the number of refinements of our tool in the experiments.
The point (50, 644) indicates that 644 examples can be solved in less than 50 refinements.
Fig.~\ref{fig:subfig:14a} shows that most of the examples can be solved in less than 80 refinements.
In our method, we successfully decomposed the complex verification problem into dozens of small problems.

\paragraph{Size of refinements.}
Our experimental results reveal that without the \tsc, the formula size reduces to 1/8 on average and to 1/1200 in the extreme case, thereby allowing for the abstractions to be solved quickly.
However, the number of clauses increased during the refinement process can usually be ignored.
Most of the examples in our experiments show hundreds or even thousands of increase in the number of CNF clauses during the refinements.
However, the CNF clause number of the abstraction may reach millions.

\paragraph{Cost of constraint solving.}
Without the \tsc, the abstractions can usually be solved instantly.
Meanwhile, the graph-based \eog validation and refinement generation processes are not trivial.
We have observed that in our experiments, our tool has spent most of its time on graph analysis for those examples where the \tsc dominates the encoding.
Meanwhile, for those examples where the complexity mainly stems from the complex data structures and numerical calculation, our tool has spent most of its time on constraint solving.

\begin{figure}
	\centering
	\subfigure[Number of refinements]{
		\label{fig:subfig:14a} 
		\includegraphics[width=5cm]{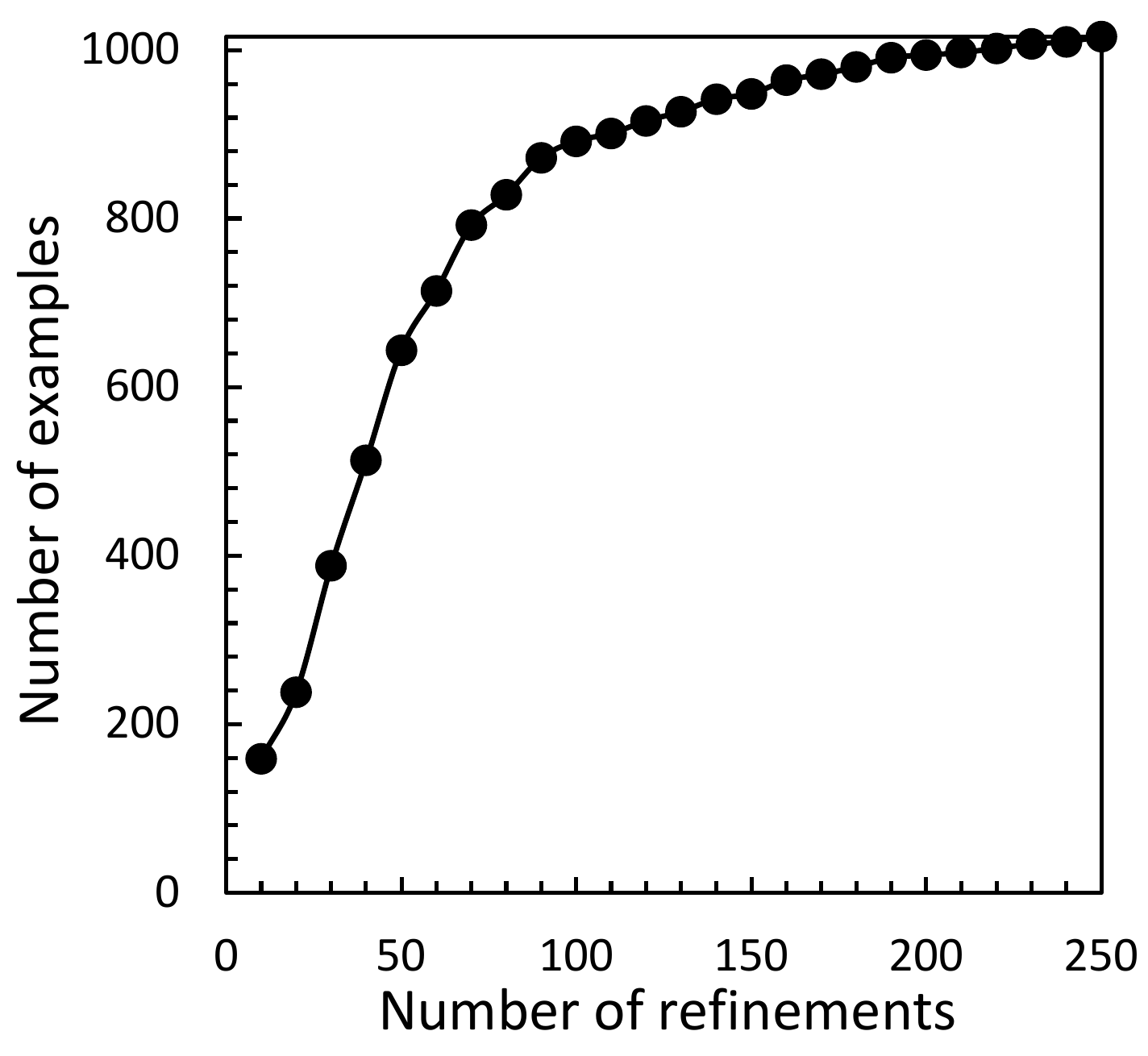}
	}
	\hspace{0.01in}
	\subfigure[Our tool versus \sat-\cegar]{
		\label{fig:subfig:14b} 
		\includegraphics[width=5cm]{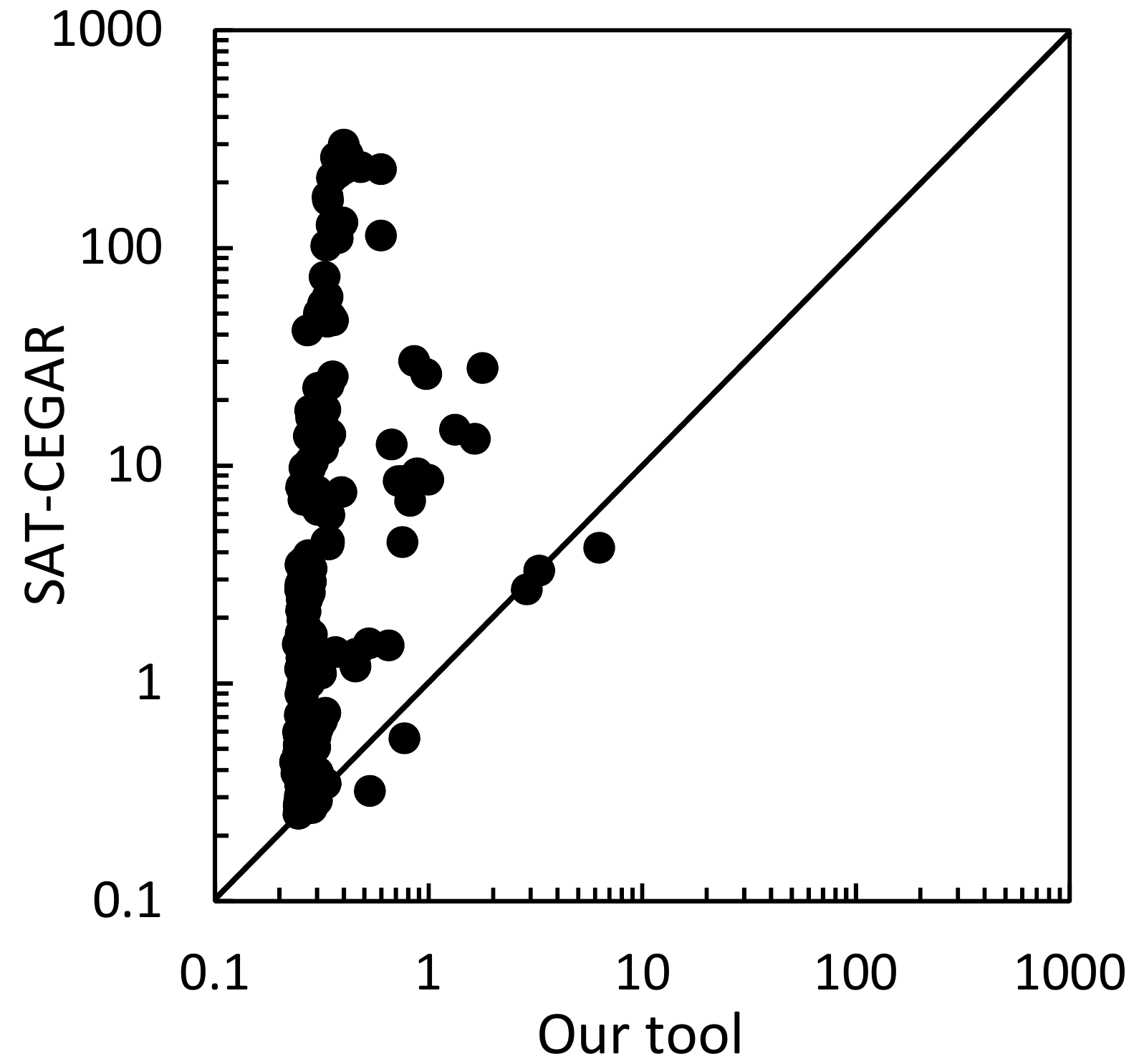}
	}
	\caption{Number of refinements and benefits from \eog}
	\label{fig:14} 
\end{figure}

\paragraph{Efficiency benefit from the graph-based refinement generation.}
The efficiency of our method mainly benefits from the graph-based refinement generation.
We have implemented another \tsc based abstraction refinement method, SAT-CEGAR, which employs only \emph{constraint solving} for \eog validation and refinement generation.
\sat-\cegar has solved only 265 of the 1047 examples in the experiments.
Fig.~\ref{fig:subfig:14b} compares the performance of our tool with \sat-\cegar in solving these examples.
Our tool outperforms \sat-\cegar for most examples, and runs 58 times faster on average.
On average, our tool finds 9.3 core kernel reasons in each refinement, with each kernel reason having an average clause length of 3.06.
Meanwhile, \sat-\cegar only finds one kernel reason in each refinement, with each reason having an average clause length of 7.5.

\paragraph{Exceptions where the graph-based \eog validation method fails}
In our method, if the constraint solving based refinement process is invoked and returns UNSAT, then we cannot achieve an effective refinement.
How often does this case occur in the experiments?
Fortunately, we have not observed cases similar to the ``butterfly'' example in our experiments, and the constraint solving based refinement process has never been invoked.

\subsection{Threats to Validity}

One threat to the validity is the limited benchmarks we have used.
For those examples where the scheduling constraint is not a major part of the encoding, our method may still need dozens of refinements.
Given that those abstractions may have similar size with the monolithic encoding, our method may perform worse than the monolithic method.

Another threat to the validity is the tools we have used.
Given that different technologies have different advantages, it is difficult to give an absolute fair comparison.
In some other scenarios, one may prefer \threader and other tools.

The third threat to the validity is the parameters we have used for each tool.
Most tools have some parameters related with their techniques, such as the loop unwinding depth.
Given that different tools may have different parameters, it is difficult to compare all the tools under the same parameters.

\section{Related Work}
\label{sec:relatedWork}

Addressing the control state explosion resulting from concurrency poses a significant challenge to concurrent program verification.
Several techniques have recently been studied to overcome this problem, including stateless model checking
\cite{Huang15,AbdullaAJS14,BarnatBHHKLRSW13,CoonsMM13}, compositional reasoning
\cite{DinsdaleBGPY13,PopeeaRW14,GuptaPA11,MalkisPR10}, bounded model checking \cite{AlglaveKT13,CordeiroMNF12,InversoT0TP14,TomascoI0TP15,ZhengRLDS15,GuntherLW16}, and abstraction refinement \cite{DanMVY15,MalkisPR10,SinhaW11,DanMVY13,GuptaHRST15}, etc.

The general idea of stateless model checking is to employ partial order reduction (POR) or dynamic partial order reduction (DPOR) \cite{AbdullaAJS14,BarnatBHHKLRSW13,CoonsMM13,ZhangKW15} to explore only non-redundant interleavings.
There are also some work which reduces the search space by restricting the schedules of the program \cite{BerganCG13,WuTHCY12}.
In compositional reasoning, rather than considering all possible interleavings of a program, the property is decomposed into different components.
Each component is then considered in isolation, without any knowledge of the precise concurrent context.
Recent work on compositional reasoning includes assume-guarantee reasoning \cite{ElkaderGPS16,ElkaderGPS15}, rely-guarantee reasoning \cite{GavranNKMV15,LahavV15,GuptaPA11}, thread-modular reasoning \cite{MalkisPR10}, and compositional reasoning \cite{DinsdaleBGPY13,PopeeaRW14}, etc.

Our \tsc based abstraction refinement method explores both bounded model checking and abstraction refinement. Afterward, we compare our study with the recent work on these two methods.

\paragraph{On Bounded Model Checking}
Bounded model checking has been considered an efficient technique to address the interleaving problem.
In \svcompc, 16 out of the 18 participants in the concurrency track have adopted this technique \cite{Beyer17}.
However, pure BMC is still not efficient enough.
Many existing tools combine this method with other techniques.
\textsc{ESBMC} combines symbolic model checking with explicit state space exploration \cite{CordeiroMNF12}.
\textsc{VVT} employs a CTIGAR method, an SMT-based IC3 algorithm that incorporates CEGAR \cite{GuntherLW16}.
The interleaving problem can also be addressed by translating the concurrent programs into sequential programs.
Tools implementing this technique include \mucseq~\cite{TomascoI0TP15}, \lazycseq~\cite{InversoT0TP14}, and \textsc{SMACK}\cite{RakamaricE14}, etc.

However, all of these work gives an exact encoding of the scheduling constraint, while we ignore this constraint and employ a \tsc based abstraction refinement method to obtain a small yet effective abstraction w.r.t. the property.

\paragraph{On Abstraction Refinement}
Abstraction refinement has been widely studied in concurrent program verification.
Most of these work employs predicate abstraction to address the data space explosion problem~\cite{GuptaPA11,DanMVY13,GuptaHRST15,DanMVY15,ZhangMGNY14,GuptaPR11}.
In predicate abstraction, it uses a finite number of predicates to abstract the program.
If an abstraction counterexample is spurious, it finds predicates that add more details of the program to refine the abstraction, s.t. the spurious counterexample is absent in the latter abstraction models.
To find the right set of predicates in less iterations, many heuristics have been proposed.
For example, Ashutosh Gupta and Thomas A. Henzinger et al. accelerated the search for the right predicates by exploring the bad abstraction traces \cite{GuptaHRST15}.
By contrast, we employ abstraction refinement to address the control space explosion problem resulting from the thread interleavings.
Our abstraction and refinement methods are both different with that of predicate abstraction.

The work most related to ours focuses on interference abstraction (IA) \cite{SinhaW11}.
N. Sinha and C. Wang also performed abstraction refinement to deal with the overhead of the exact encoding of the concurrent behavior.
However, they abstracted the behavior by restricting the sets of read events and read-write links, while we consider all read events and read-write links but relax the \tsc.
Accordingly, their abstraction was refined by introducing new read events and read-write links, while we perform the refinement by exploring a graph-based method to analyze core kernel reasons that make an counterexample infeasible.
Moreover, they employ a mixed framework of over- and under-approximations, while our method produces only over-approximation abstractions.
Given that their implementation was for Java program slices, an empirical comparision between their and our method is difficult.

Another work closely related to ours is \cite{Kusano016}.
In this work, M. Kusano and C. Wang also presented a set of deduction rules to help determine the infeasibility of an interference combination. However, our task is to determine the feasibility of a counterexample which contains a large number of read-write links, and our main innovation is to devise a graph-based refinement generation method to obtain an effective refinement constraint.
In addition, our deduction rules are much simpler yet stronger than theirs.

To deal with the interleaving problem, A. Farzan and Z. Kincaid also divided the verification into data and control modules, and incorporated them into an abstraction refinement framework \cite{FarzanK12,FarzanK13}. The difference is that in their work, the verification is reasoned by data-flow analysis, while we represent the program by SSA statements and employ graph and constraint based EOG analysis approaches to do the refinement.
In addition, their work focuses on parameterized programs, while we concentrate on multi-threaded programs based on PThreads.

\section{Conclusions}
\label{sec:conclusion}

This paper proposed a scheduling constraint based abstraction refinement method for multi-threaded program verification.
To obtain an effective refinement, we also devised two graph-based algorithms for counterexample validation and refinement generation. Our experiment results on benchmarks of \svcompc show that our method is promising and significantly outperforms the existing state-of-the-art tools.
We plan to extend this technique to weak memory models, such as TSO, PSO, POWER, in the future.


\end{document}